\documentclass{article}
\usepackage{graphicx} 

\usepackage[english]{babel}
\usepackage[utf8]{inputenc}
\usepackage[T1]{fontenc}
\usepackage[a4paper, margin=1in]{geometry}
\usepackage{graphicx}
\usepackage{framed}
\usepackage[framemethod=tikz]{mdframed}
\usepackage{todonotes}
\usepackage{tablefootnote}
\usepackage{color}

\newcommand{\alp}{\alpha}
\newcommand{\eps}{\epsilon}
\newcommand{\Omg}{\Omega}
\newcommand{\lmax}{\ell_{\max}}
\newcommand{\tmax}{t_{\max}}
\newcommand{\hmin}{h_{\min}}
\newcommand{\taumax}{\tau_{\max}}
\newcommand{\calP}{\mathcal{P}}
\newcommand{\bundle}{\text{pass-bundle}\xspace}
\newcommand{\limit}{\mathsf{limit}}

\newcommand{\lab}{\mathsf{distance}}
\newcommand{\algPhase}{\textsc{Alg-Phase}\xspace}

\newcommand{\algExtend}{\textsc{Extend-Active-Path}\xspace}
\newcommand{\algBacktrack}{\textsc{Backtrack-Stuck-Structures}\xspace}
\newcommand{\algOvertake}{\textsc{Overtake}\xspace}
\newcommand{\algAugment}{\textsc{Augment}\xspace}
\newcommand{\algContract}{\textsc{Contract}\xspace}
\newcommand{\algCheck}{\textsc{Contract-and-Augment}\xspace}

\newcommand\cev[1]{\overleftarrow{#1}}

\definecolor{darkgreen}{rgb}{0,0.5,0}
\definecolor{darkgray}{rgb}{0.2,0.2,0.2}
\usepackage[backref=page]{hyperref}
\hypersetup{
    unicode=false,          
    colorlinks=true,        
    linkcolor=blue,         
    citecolor=purple,       
    filecolor=magenta,      
    urlcolor=cyan           
}

\usepackage{amsthm}
\usepackage{amsmath}
\usepackage{amssymb}
\usepackage{amsfonts}
\usepackage{mathrsfs}
\usepackage{mathtools}
\usepackage{verbatim}
\usepackage{footnote}

\usepackage{algorithm}
\usepackage[noend]{algpseudocode}
  
\usepackage{lineno}
\usepackage{caption}
\usepackage{framed}
\usepackage{enumerate}

\usepackage{tikzsymbols}
\usepackage{thmtools,thm-restate}
\usepackage{nicefrac}

\usepackage{subcaption}
\usepackage{pdfpages}
\usepackage{makecell}

\usepackage[capitalize, nameinlink]{cleveref}
\usepackage[section]{placeins}
\crefname{theorem}{Theorem}{Theorems}
\Crefname{lemma}{Lemma}{Lemmas}
\Crefname{invariant}{Invariant}{Invariants}
\Crefname{claim}{Claim}{Claims}
\Crefname{observation}{Observation}{Observations}
\Crefname{algorithm}{Algorithm}{Algorithms}
\Crefname{figure}{Figure}{Figures}
\Crefname{challenge}{Challenge}{Challenges}

\newtheorem{theorem}{Theorem}[section]
\newtheorem{lemma}[theorem]{Lemma}

\newtheorem{definition}[theorem]{Definition}
\newtheorem{invariant}[theorem]{Invariant}

\newtheorem{remark}{Remark}
\newtheorem*{remark*}{Remark}

\DeclareMathOperator{\poly}{poly}

\def\LOCAL{\ensuremath{\mathsf{LOCAL}}\xspace}
\def\CONGEST{\ensuremath{\mathsf{CONGEST}}\xspace}

\newcommand{\eMM}{\eps\textsc{MM}}

\newcommand{\mytab}{\ \ \ \ }
\newcommand{\rb}[1]{\left( #1 \right)}

\title{
Faster Semi-streaming Matchings via Alternating Trees
}
\author{ Slobodan Mitrović\thanks{Supported by the Google Research Scholar and NSF Faculty Early Career Development Program No.~2340048. e-mail: \texttt{smitrovic@ucdavis.edu}} \\ UC Davis 
\and Anish Mukherjee\thanks{e-mail: \texttt{Anish.Mukherjee@liverpool.ac.uk}} \\ University of Liverpool
\and Piotr Sankowski\thanks{Supported by National Science Center grant no.~2020/37/B/ST6/04179. e-mail: \texttt{sank@mimuw.edu.pl}} \\ University of Warsaw
\and Wen-Horng Sheu\thanks{e-mail: \texttt{wsheu@ucdavis.edu}} \\ UC Davis}
\date{ }

\begin{document}

\maketitle

\begin{abstract}
    We design a deterministic algorithm for the $(1+\epsilon)$-approximate maximum matching problem. 
    Our primary result demonstrates that this problem can be solved in $O(\epsilon^{-6})$ semi-streaming passes, improving upon the $O(\epsilon^{-19})$ pass-complexity algorithm by [Fischer, Mitrović, and Uitto, STOC'22]. 
    This contributes substantially toward resolving Open question~2 from [Assadi, SOSA'24].
    Leveraging the framework introduced in [FMU'22], our algorithm achieves an analogous round complexity speed-up for computing a $(1+\epsilon)$-approximate maximum matching in both the Massively Parallel Computation (MPC) and CONGEST models.

    The data structures maintained by our algorithm are formulated using blossom notation and represented through alternating trees.
    This approach enables a simplified correctness analysis by treating specific components as if operating on bipartite graphs, effectively circumventing certain technical intricacies present in prior work.

\end{abstract}

\newpage

\tableofcontents

\newpage

\section{Introduction}


Given an undirected, unweighted graph $G = (V, E)$ of $n$ vertices and $m$ edges, the task of maximum matching is to find the largest set of edges $M \subseteq E$ such that no two edges in $M$ share an endpoint. 
With the prominence of large volumes of data and huge graphs, there has been a significant interest in finding simple and very fast algorithms, even at the expense of allowing approximation in the output.
In particular, given a constant $\eps > 0$, the problem of finding an $(1+\eps)$-approximate maximum matching (that we denote by $\eMM$) has been studied in the \emph{semi-streaming model}~\cite{mcgregor2005finding,ahn2011laminar,ahn2013linear,tirodkar2018deterministic,ahn2018access,gamlath2019weighted,FMU22,assadi2022semi,assadi2024simple,huang2023}.

In the semi-streaming setting, the algorithm is assumed to have access to $O(n \poly \log n)$ space, which is generally insufficient to store the entire input graph.
The graph $G$ is provided as a stream of $m$ edges arriving in an arbitrary order.
The objective is to solve the problem with the minimum possible number of passes over the stream.
The semi-streaming model has been the model of choice for processing massive graphs, as the traditional log-space streaming model proves too restrictive for many fundamental graph problems.
In particular, it has been shown that testing graph connectivity requires $\Omg(n)$ space in the streaming setting, even with a constant number of passes allowed \cite{10.5555/327766.327782}.

In the study of $\eMM$ within the semi-streaming setting, one of the primary aims is to deliver methods whose dependence on $1/\eps$ in the pass complexity is as small as possible while retaining the smallest known dependence on $n$.
There has been ample success in this regard when the input graph is bipartite, where $O(1/\eps^2)$-pass algorithm is known \cite{assadi2021auction}.
We note that the pass complexity of this algorithm does not depend on $n$, but only a polynomial of $1/\eps$.

For general graphs, the situation is very different. 
For instance, until very recently, the best pass complexity in the semi-streaming setting either depends exponentially on $1/\eps$~\cite{mcgregor2005finding,tirodkar2018deterministic,gamlath2019weighted}, or polynomially on $1/\eps$ but with a dependence on $\log n$~\cite{ahn2013linear,ahn2011laminar,ahn2018access,assadi2022semi,assadi2024simple}; we provide an overview of existing results in \cref{table:runningtimes,table:runningtimes-general}. 
Significant progress has been made by Fischer, Mitrovi\' c, and Uitto~\cite{FMU22}, who introduced a semi-streaming algorithm for general graphs that outputs a $(1+\eps$)-approximate maximum matching in $\poly(1/\eps)$ passes, i.e., $O(1/\eps^{19})$ many passes with no dependence on $n$.
Also, in that work, improvements of the same quality were obtained for the MPC and \CONGEST models, albeit with a higher polynomial dependence on $1/\eps$.

Although the result \cite{FMU22} makes important progress in understanding semi-streaming algorithms for approximating maximum matchings in general graphs, the analysis presented in that work is quite intricate.
In addition, the gap between known complexities for methods tackling bipartite and general graphs in semi-streaming remains relatively large, i.e., $1/\eps^2$ vs. $1/\eps^{19}$. This inspires the main question of our work:
\begin{center}
    \emph{Can we design simpler and more efficient semi-streaming algorithms \\ for $(1+\eps$)-approximate maximum matching in general graphs?}
\end{center}

\subsection{Our contribution}
\label{sec:contribution}
The main technical claim of our work can be summarized as follows.
\begin{restatable}{theorem}{maintheorem}
\label{thm:main}
There exists a deterministic semi-streaming algorithm that, given any $\eps > 0$ and a graph $G$ on $n$ vertices, outputs a $(1+\eps)$-approximate maximum matching in $G$ in $O(1/\eps^6)$ passes. The algorithm uses $O(n /\eps^6)$ words of space.
\end{restatable}
The very high-level idea behind our approach is finding ``short'' augmentations until only a few of them remain. It is folklore that such an algorithm yields the desired approximation.
To find these ``short'' augmentations, each free vertex maintains some set of alternating paths originating at it;
we refer to such a set of alternating paths by a \emph{structure}.
The main conceptual contribution of our work is representing these structures by \emph{alternating trees} and \emph{blossoms}, introduced by Edmonds in his celebrated work~\cite{edmonds1965paths}.
{An alternating tree is a tree rooted at a free vertex such that every root-to-leaf path is an alternating path.
A formal, recursive definition of blossoms is given in \cref{def:blossom}.
Informally, a blossom can be either an odd alternating cycle (that is, an odd cycle in which each vertex except one is matched to one of its neighbors in the cycle), or a set of smaller vertex-disjoint blossoms connected by an odd alternating cycle.
Our algorithm represents each structure by an alternating tree in which each node can be either a vertex (of the input graph) or a blossom contracted into a single vertex.}
The most related work to ours, i.e., \cite{FMU22}, develops an ad-hoc structure. Although another related work \cite{huang2023} builds on blossom structure, it does it for reasons other than finding short augmentations; in fact, for finding short augmentations \cite{huang2023} uses \cite{FMU22} essentially in a black-box manner.
A more detailed overview of prior work is given in \cref{table:runningtimes,table:runningtimes-general}.

Our approach provides several advantages:
(1) we re-use some of the well-established properties of blossoms;
(2) our proof of correctness and our algorithm are simpler compared to that of \cite{FMU22}; (3) we exhibit relations between different alternating trees during the short-augmentation search that eventually lead to our improved pass complexity, from $1 / \eps^{19}$ to $1 / \eps^{6}$;
and (4) the size of structures in our algorithm is considerably smaller than the structure size in \cite{FMU22}'s algorithm, which leads to the improvement of space complexity from $O(n/\eps^{22})$ to $O(n/\eps^{6})$.
We hope that this perspective and simplification in the analysis will lead to further improvements in designing approximate maximum matching algorithms.


\begin{table}[h]
\centering
\setlength{\extrarowheight}{4pt}
\begin{tabular}{|c|c|c|}
\hline
Reference & Passes & Deterministic \\ \hline
\cite{eggert2009bipartite} & $O(1/\eps^8)$ & Yes \\ \hline
\cite{eggert2012bipartite} & $O(1/\eps^5)$ & Yes \\ \hline
\cite{ahn2013linear} & $O(1/\eps^2 \cdot \log (1/\eps))$ & Yes \\ \hline
\cite{kapralov2013better} & $O(1/\eps^2)$ (vertex arrival) & Yes \\ \hline
\cite{ahn2018access} & $O(\log{(n)}/\eps)$  & No \\ \hline
\cite{assadi2021auction} & $O(1/\eps^2)$ & Yes \\ \hline
\cite{assadi2022semi} & $O(\log(n)/\eps \cdot \log (1/\eps))$ & Yes \\ \hline
\cite{assadi2024simple} & $O(\log{(n)}/\eps)$  & Yes \\ \hline
\end{tabular}

\caption{\label{table:runningtimes} A summary of the pass complexities of computing $(1+\eps)$-approximate maximum matching in \textbf{bipartite graphs}. Each algorithm uses $O(n \cdot \poly(\log n, 1/\eps))$ space, although some have tighter guarantees.
Using the framework in \cite{Bernstein2021} and its follow-up \cite{Bernstein2025}, all results in this table also hold for the weighted case, except that the $\log n$ factors are increased to $\log(n / \eps)$.
}
\end{table}

\begin{table}[h]
\centering
\setlength{\extrarowheight}{4pt}
\begin{tabular}{|c|c|c|c|c|}
\hline
Reference & Passes & Deterministic & Weighted \\ \hline
\cite{mcgregor2005finding} & $\exp(1/\eps)$  & No & No \\ \hline
\cite{ahn2011laminar} & $O(\log{(n)}/\eps^7 \cdot \log (1/\eps))$  & Yes & Yes \\ \hline
\cite{ahn2013linear} & $O(\log{(n)}/\eps^4)$  & Yes & Yes \\ \hline
\cite{ahn2018access} & $O(\log{(n)}/\eps)$  & No & Yes \\ \hline
\cite{tirodkar2018deterministic} & $\exp(1/\eps)$  & Yes & No \\ \hline
\cite{gamlath2019weighted} & $\exp(1/\eps^2)$  & No & Yes \\ \hline
\cite{FMU22} & $O(1/\eps^{19})$  & Yes & No \\ \hline
\cite{huang2023} & more than $O(1/\eps^{19})$  & Yes & Yes \\ \hline
\cite{assadi2024simple} & $O(\log{(n)}/\eps)$  & Yes & Yes \\ \hline
\hline
\textbf{this work} & $O(1/\eps^{6})$  & Yes & No \\ \hline
\end{tabular}

\caption{\label{table:runningtimes-general} A summary of the pass complexities of computing $(1+\eps)$-approximate maximum matching in \textbf{general graphs}. Each algorithm uses $O(n \cdot \poly(\log n, 1/\eps))$ space, although some have tighter guarantees.
\vspace{-0pt}
}
\end{table}


\paragraph{Implication to other models.}
The $\eps$MM problem has been comprehensively studied in other models of computation as well, such as Massively Parallel Computation (MPC)~\cite{czumaj2018round,assadi2019coresets,ghaffari2018improved,ghaffari2022massively,assadi2021auction}, \CONGEST and \LOCAL~\cite{lotker2015improved,bar2017distributed,fischer2017deterministic,ahmadi2018distributed,ghaffari2018derandomizing,ghaffari2018deterministic,harris2019distributed,faour2021distributed}.
\cite{FMU22} provides a framework which, through their semi-streaming algorithm, reduces the computation of $(1+\eps)$-approximate maximum matching to $\poly(1/\eps)$ invocations of a $\Theta(1)$-approximate maximum matching algorithm. A straightforward adaptation of their framework to our result yields $1/\eps^{13}$ round-complexity improvement for MPC and \CONGEST.

\subsection{Related work} \label{sec:related-work}
Our algorithmic setup is inspired by \cite{FMU22} and its predecessor \cite{eggert2012bipartite}, while several structural properties are borrowed from \cite{edmonds1965paths}. We detail similarities and differences between our and \cite{FMU22}'s techniques in \cref{sec:comparison-with-FMU}.

Approximate maximum matchings in (semi-)streaming have been extensively studied from numerous perspectives.
The closest to our work is \cite{FMU22}, who design a deterministic semi-streaming algorithm for $\eMM$ using $O(1/\eps^{19})$ passes. Prior to that work, \cite{mcgregor2005finding,tirodkar2018deterministic} developed semi-streaming algorithms that use $\exp(1/\eps)$ many passes. \cref{table:runningtimes,table:runningtimes-general} list many other results for computing $\eMM$ in semi-streaming; in bipartite and general graphs, respectively.

Next, we briefly describe some related works on variants of the $\eMM$ problem or the underlying model of computation, that have been considered in the literature. 
A sequence of papers has studied the question of estimating the size of the maximum matching \cite{KapralovKS14,BuryS15,AssadiKL17,EsfandiariHLMO18,KapralovMNT20}. The $\eMM$ problem has been considered in weighted graphs as well \cite{ahn2011laminar,ahn2013linear,BuryS15,ahn2018access, gamlath2019weighted,huang2023,assadi2024simple}. 
On bipartite graphs and in the semi-streaming model, any algorithm for (unweighted) $\eMM$ can be converted to an algorithm for weighted $\eMM$ with slightly increased pass and space complexities \cite{Bernstein2021,Bernstein2025}.
Using this conversion, all pass-complexity results in \cref{table:runningtimes} also apply to the weighted case, except that the $\log n$ factors are increased to $\log(n / \eps)$.
Considering variants of the streaming setting, there have been works in dynamic streaming where one can both insert and remove edges \cite{Konrad15, BuryS15, ChitnisCEHMMV16,AssadiKLY16}, in the vertex arrival model \cite{KarpVV90,GoelKK12,kapralov2013better,EpsteinLSW13, ChiplunkarTV15,BuchbinderST19,GamlathKMSW19}, and in random streaming where vertices or edges arrive in a random order \cite{MahdianY11,KonradMM12,gamlath2019weighted,FarhadiHMRR20,Bernstein20,AssadiB21}. 

Several works have studied lower bound questions in the streaming setting, both for exact \cite{GuruswamiO16,AssadiR20,ChenKPS0Y21} and approximate maximum matching \cite{GoelKK12,kapralov2013better,AssadiKLY16,AssadiKL17,AssadiKSY20,Kapralov21,Assadi22,AssadiS23}. 

The $\eMM$ problem is well-studied in the classical centralized setting as well \cite{DrakeH03, DP14}.
Furthermore, in recent years, there has been a growing interest in $\eMM$ in the area of dynamic algorithms \cite{GLSSS19,BGS20, ABD22,AssadiBKL23,BlikstadK23,ZhengH23,BhattacharyaKS23,behnezhad2024fully,assadi2024improved} and also in sublinear time algorithms for approximate maximum matching \cite{Behnezhad21,BehnezhadRR23,BehnezhadRRS23}. 

\subsection{Organization}
After some preliminaries in \cref{sec:preliminaries}, we give an overview of our approach in \cref{sec:overview}. Next, in \cref{sec:alg}, we describe our algorithm in detail. We argue about the correctness of our algorithm in \cref{sec:correctness} followed by the proof of the claimed complexity bounds of our algorithm in \cref{sec:pass-complexity}.


\section{Preliminaries}\label{sec:preliminaries}
In the following, we first introduce all the terminology, definitions, and notations. We also recall some well-known facts about blossoms.

Let $G$ be an undirected simple graph and $\eps \in (0, 1]$ be the approximation parameter.
Without loss of generality, we assume that $\eps^{-1}$ is a power of $2$.
Denote by $V(G)$ and $E(G)$, respectively, the vertex and edge sets of $G$.
Let $n$ be the number of vertices in $G$ and $m$ be the number of edges in $G$.
An undirected edge between two vertices $u$ and $v$ is denoted by $\{u, v\}$.
Throughout the paper, if not stated otherwise, all the notations implicitly refer to a currently given matching $M$, which we aim to improve.

\subsection{Alternating paths}

\begin{definition}[An unmatched edge and a free vertex] We say that an edge $\{u, v\}$ is \emph{matched} iff $\{u, v\} \in M$, and \emph{unmatched} otherwise.
We call a vertex $v$ \emph{free} if it has no incident matched edge, i.e., if $\{u, v\}$ are unmatched for all edges $\{u, v\}$.
Unless stated otherwise, $\alp, \beta, \gamma$ are used to denote free vertices. 
\end{definition}

\begin{definition}[Alternating and augmenting paths] An \emph{alternating path} is a simple path that consists of a sequence of alternately matched and unmatched edges. The \emph{length} of an alternating path is the number of edges in the path. An \emph{augmenting path} is an alternating path whose two endpoints are both free vertices.
\end{definition}

\subsection{Alternating trees and blossoms}

\begin{definition}[Alternating trees, inner vertices, and outer vertices]
\label{def:alternating-tree}
A subgraph of $G$ is an \emph{alternating tree} if it is a rooted tree in which the root is a free vertex and every root-to-leaf path is an even-length alternating path.
An \emph{inner vertex} of an alternating tree is a non-root vertex $v$ such that the path from the root to $v$ is of odd length.
All other vertices are \emph{outer vertices}.
In particular, the root vertex is an outer vertex.
\end{definition}

Note that in an alternating tree, every leaf is an outer vertex; every inner vertex $v$ has exactly one child, which is matched to $v$. Hence, every non-root vertex in the tree is matched.

\begin{definition}[Blossoms and trivial blossoms]
\label{def:blossom}
A blossom is identified with a vertex set $B$ and an edge set $E_B$ on $B$.
If $v \in V(G)$, then $B = \{v\}$ is a \emph{trivial blossom} with $E_B = \emptyset$.
Suppose there is an odd-length sequence of vertex-disjoint blossoms $A_0, A_1, \dots , A_k$ with associated edge sets $E_{A_0}, E_{A_1}, \dots, E_{A_k}$.
If $\{A_i\}$ are connected in a cycle by edges $e_0, e_1, \dots , e_k$, where $e_i \in A_i \times A_{i+1} (\mbox{modulo } k+1)$ and $e_1, e_3, \dots, e_{k-1}$ are matched, then $B = \bigcup_{i} A_i$ is also a blossom associated with edge set $E_B = \bigcup_i E_{A_i} \cup \{e_0, e_1, \dots , e_k\}$.
\end{definition}

\begin{remark}
In the literature, a blossom is often defined as an odd-length cycle in $G$ consisting of $2k+1$ edges, where exactly $k$ of these edges belong to the matching $M$.
Here, we use the definition in \cite{DP14}, in which a blossom is defined recursively as an odd-length cycle alternating between matched and unmatched edges, whose components are either single vertices or blossoms in their own right.
This recursive definition characterizes the subgraphs contracted in Edmond's algorithm.
\end{remark}

Consider a blossom $B$.
A short proof by induction shows that $|B|$ is odd.
In addition, $M \cap E_B$ matches all vertices except one.
This vertex, which is left unmatched in $M \cap E_B$, is called the \emph{base} of $B$.
Note that $E(B) = E(G) \cap (B \times B)$ may contain many edges outside of $E_B$.
Blossoms exhibit the following property.

\begin{lemma}[\cite{DP14}]\label{lem:even_path} Let $B$ be a blossom. There is an even-length alternating path in $E_B$ from the base of $B$ to any other vertex in $B$. \end{lemma}

\begin{definition}[Blossom contraction] Let $B$ be a blossom. We define the contracted graph $G / B$ as the undirected simple graph obtained from $G$ by contracting all vertices in $B$ into a vertex, denoted by $B$.
\end{definition}

The following lemma is proven in \cite[Theorem 4.13]{edmonds1965paths}.

\begin{lemma}[\cite{edmonds1965paths}]\label{lem:contraction} Let $T$ be an alternating tree of a graph $G$ and $e \in E(G)$ be an edge connecting two outer vertices of $T$. Then, $T \cup \{e\}$ contains a unique blossom $B$. The graph $T / B$ is an alternating tree of $G / B$. It contains $B$ as an outer vertex. Its other inner and outer vertices are those of $T$ which are not in $B$.
\end{lemma}

Consider a set $\Omg$ of blossoms.
We say $\Omg$ is \emph{laminar} if the blossoms in $\Omg$ form a laminar set family.
Assume that $\Omg$ is laminar.
A blossom in $\Omg$ is called a \emph{root blossom} if it is not contained in any other blossom in $\Omg$.
Denote by $G / \Omg$ the undirected simple graph obtained from $G$ by contracting each root blossom of $\Omg$.
For each vertex in $\bigcup_{B \in \Omg} B$, we denote by $\Omg(v)$ the unique root blossom containing $v$.
If $\Omg$ contains all vertices of $G$, we denote by $M / \Omg$ the set of edges $\{ \{\Omg(u), \Omg(v)\} \mid \{u, v\} \in M \mbox{ and } \Omg(u) \neq \Omg(v) \}$ on the graph $G / \Omg$. It is known that $M / \Omg$ is a matching of $G / \Omg$ \cite{DP14}.

In our algorithm, we maintain several vertex-disjoint subgraphs. Each subgraph is associated with a \emph{regular} set of blossoms, which is a set of blossoms whose contraction would transform the subgraph into an alternating tree satisfying certain properties. A regular set of blossoms is formally defined as follows.

\begin{definition}[Regular set of blossoms]
A regular set of blossoms of $G$ is a set $\Omg$ of blossoms satisfying the following:
\begin{itemize}
    \item[(C1)] $\Omg$ is a laminar set of blossoms of $G$. It contains the set of all trivial blossoms in $G$. 
    If a blossom $B \in \Omg$ is defined to be the cycle formed by $A_0, \dots, A_k$, then $A_0, \dots, A_k \in \Omg$.
    \item[(C2)] $G / \Omg$ is an alternating tree with respect to the matching $M / \Omg$. Its root is $\Omg(\alp)$ and each of its inner vertex is a trivial blossom (whereas each outer vertex may be a non-trivial blossom).
\end{itemize}    
\end{definition}



\subsection{Representation of undirected graphs}

In our algorithm, each undirected edge $\{u, v\}$ is represented by two directed \emph{arcs} $(u, v)$ and $(v, u)$.
Let $(u, v)$ be an arc.
We say $(u, v)$ is \emph{matched} if $\{u, v\}$ is a matched edge;
otherwise, $(u, v)$ is \emph{unmatched}.
The vertex $u$ and $v$ are called, respectively, \emph{tail} and \emph{head} of $(u, v)$.
We denote by $\cev{(u, v)} = (v, u)$ the reverse of $(u, v)$.


Let $P = (u_1, v_1, \dots, \allowbreak u_k, v_k)$ be an alternating path, where $u_i$ and $v_i$ are vertices, $(u_i, v_i)$ are matched arcs, and $(v_i, u_{i+1})$ are unmatched ones.
Let $a_i = (u_i, v_i)$.
We often use $(a_1, a_2, \dots, a_k)$ to refer to $P$, i.e., we omit specifying unmatched arcs.
Nevertheless, it is guaranteed that the input graph contains the unmatched arcs $(v_i, u_{i+1})$, for each $1 \leq i < k$.
If $P$ is an alternating path that starts and/or ends with unmatched arcs, e.g., $P = (x, u_1, v_1, \dots, u_k, v_k, y)$ where $(x, u_1)$ and $(v_k, y)$ are unmatched while $a_i = (u_i, v_i)$ for $i = 1 \dots k$ are matched arcs, we use $(x, a_1, . . . , a_k, y)$ to refer to $P$. In our case, very frequently, $x$ and $y$ will be free vertices, usually $x = \alp$ and $y = \beta$.

\begin{definition}[Concatenation of alternating paths]
Let $P_1 = (a_1, a_2, \dots , a_k)$ and $P_2 = (b_1, b_2, \dots , b_s)$ be alternating paths.
We use $P_1 \circ P_2 = (a_1, a_2, \dots, a_k, b_1, b_2, \dots , b_s)$ to denote their concatenation.
Note that, the alternating path $P_1 \circ P_2$ also contains the unmatched edge between $a_k$ and $b_1$.  
\end{definition}

\subsection{Semi-streaming model}
In the semi-streaming model~\cite{FeigenbaumKMSZ05}, we assume that the algorithm has no random access to the input graph.
The set of edges is represented as a stream.
In this stream, each edge is presented exactly once, and each time the stream is read, edges may appear in an arbitrary order.
The stream can only be read as a whole and reading the whole stream once is called a \emph{pass} (over the input).
The main computational restriction of the model is that the algorithm can only use $O(n \poly \log n)$ words of space, which is not enough to store the entire graph if the graph is sufficiently dense.


\section{Overview of Our Approach}\label{sec:overview}
The starting point of our approach is the classical idea of finding augmenting paths to improve the current matching~\cite{berge1957two,edmonds1965paths,hopcroft1971n5}.
It is well-known that it suffices to search for $O(1/\eps)$ long augmenting paths, i.e., it suffices to search for relatively short paths, to obtain a $(1+\eps)$-approximate maximum matching ($\eMM$). \emph{However, how can this short-augmentations search be performed in a small number of passes?}

To make this search efficient in the semi-streaming setting, the general idea is to search for \emph{many} augmenting paths during the same pass. This is achieved by a depth-first-search (DFS) exploration truncated at depth $O(1/\eps)$ from each free vertex; that kind of approach was employed in prior work, e.g., \cite{mcgregor2005finding,eggert2012bipartite,tirodkar2018deterministic,FMU22}.


\paragraph{Remark:} Throughout the paper, we attempt to use terminology as closely as possible to work prior, particularly the terminology used in \cite{FMU22}. 
We hope that in this way, we aid readers in comparing our contributions to priors.

\subsection{Augmentation search via alternating trees}
In our algorithm, each free vertex maintains an alternating tree. These trees are created via a DFS exploration in an alternating-path manner.
Consider an alternating tree $S$. $S$ is associated with a so-called \emph{working vertex}, which represents the last vertex the DFS exploration has currently reached. \cref{fig:alternating-tree} depicts an alternating tree.
\begin{figure}[h]
    \centering
    \includegraphics[width=0.55\textwidth]{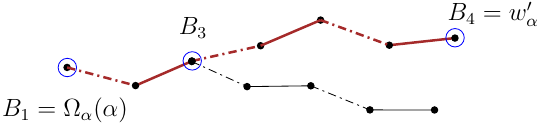}
    \caption{An example of alternating tree. Dashed and solid edges denote the unmatched and matched edges, respectively.
    The encircled vertices correspond to the non-trivial blossoms, i.e., $B_1$, $B_3$, and $B_4$ are non-trivial blossoms.
    $B_1 = \Omg_\alp(\alp)$ is the blossom containing a free vertex $\alp$.
    $w'_\alp$ is the working vertex and the highlighted path, from $B_1$ to $B_4$, is the active path.
    }
    \label{fig:alternating-tree}
\end{figure}

Recall that the goal is to look for \emph{short} augmentations. Hence, these DFS explorations are carried out by attempting to visit an edge by as short an alternating path as possible. 
In particular, each matched edge $e$ maintains a \emph{label} $\ell(e)$ representing the so-far shortest discovered alternating path to a free vertex.\footnote{Our algorithm maintains arc labels; an edge $\{u, v\}$ is represented by arcs $(u, v)$ and $(v, u)$, and the algorithm maintains $\ell((u, v))$ and $\ell((v, u))$. For the sake of simplicity, in this overview, we only talk about edge labels.}
Observe that only matched edges maintain labels. That enables storing those labels in $O(n)$ words.

In a single pass, the working vertex $u$ of $S$ attempts to extend $S$ by a length-$2$ path $\{u, v, t\}$, where $g = \{u, v\}$ is unmatched and $e = \{v, t\}$ is a matched edge; if it is impossible, just like in a typical DFS, this working vertex backtracks. Then, one of the following happens:
\begin{enumerate}
    \item The edge $g$ connects $S$ with an alternating tree $S'$, different than $S$, such that there is an augmenting path between the roots of $S$ and $S'$ involving $g$.
    In that case, this augmenting path is recorded, and $S$ and $S'$ are temporarily removed from the graph.
    This is done by procedure \algAugment, \cref{sec:augment}.
    (After the DFS exploration is completed, the algorithm restores all temporarily removed vertices and augments the matching using the recorded augmenting paths.)
    
    \item If $g$ connects two vertices in $S$ such that it creates a blossom, then this blossom is contracted. This is done by procedure \algContract, \cref{sec:contract}. These contractions ensure that the exploration subgraph from each free vertex looks like a tree.

    \item The alternating path along $S$ to $e$ is shorter than $\ell(e)$. Then, $g$ and $e$ are added to $S$, and $\ell(e)$ is updated accordingly. If $e$ belongs to another alternating \emph{subtree}, which might belong to $S$ or another alternating tree, then the entire subtree together with $e$ is appended to $u$. This is done by procedure \algOvertake, \cref{sec:overtake}.
    Note that, by the construction, the last edge appended to DFS exploration is matched.
    An example of \algOvertake is depicted in \cref{fig:overtaking-example1-overview}.
    \begin{remark}
        The intuition behind the overtake operation is as follows.
        In our algorithm, an edge label represents the shortest alternating distance from a free vertex to the edge. Thus, when $S$ finds a shorter path to $e$, $\algOvertake$ allows $S$ to take over the search on $e$ and reduce its label.
        This operation helps the algorithm find shorter alternating paths to each matched edge.
    \end{remark}
    
\end{enumerate}
\begin{figure}
\centering
    \begin{subfigure}[t]{0.45\linewidth}
        \centering
        \includegraphics[width=\textwidth]{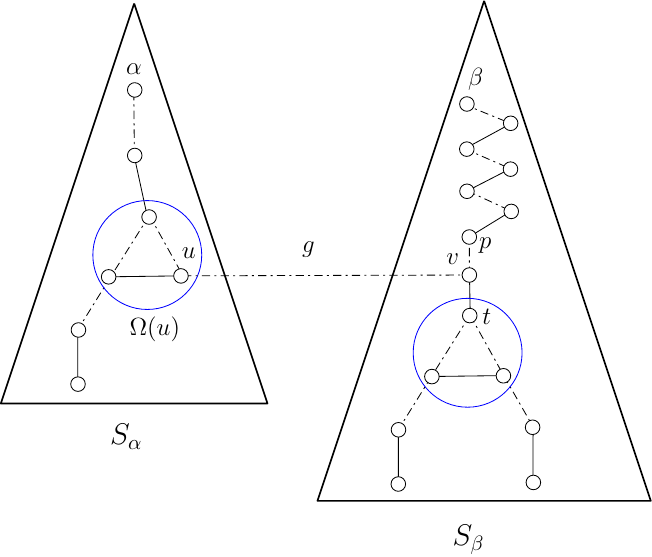}
        \caption{Before $\algOvertake$}
    \end{subfigure}
    \hfill
     \begin{subfigure}[t]{0.45\linewidth}
        \centering
        \includegraphics[width=\textwidth]{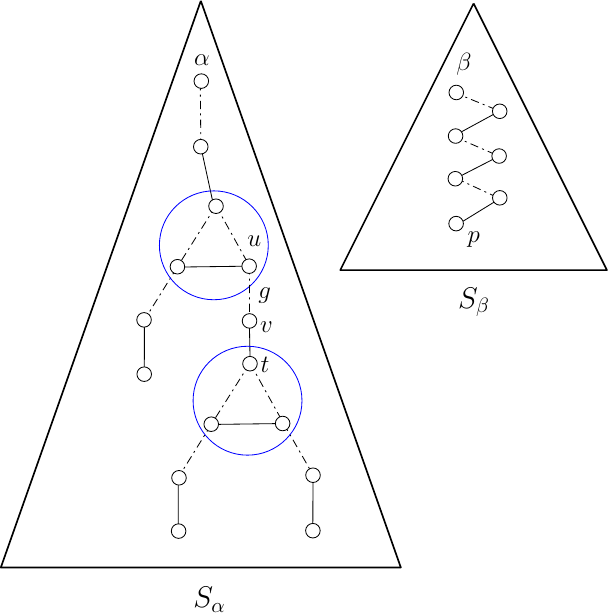}
        \caption{After $\algOvertake$}
    \end{subfigure}
    \caption{An example of $\algOvertake$. In this example, $g = (u, v)$ connects two alternating trees $S_\alp$ and $S_\beta$, with $\Omg(u)$ being the working vertex of $S_\alp$ before $\algOvertake$. The circles represent blossoms, which are not contracted in this sketch so as to illustrate possible situations better. 
    The alternating path from $\alp$ to matched edge $(\Omg(v), \Omg(t))$ along $g$ improves the label of $(\Omg(v), \Omg(t))$, and hence \algOvertake is invoked. (Observe that $\{v\} = \Omg(v)$.)
    }
    \label{fig:overtaking-example1-overview}
\end{figure}

The described operations are very natural when we take the perspective of maintaining edge labels and alternating trees.
\emph{However, why do they yield an approximate maximum matching? Moreover, how many passes does the entire process take?}

\subsection{Correctness argument (\cref{sec:correctness})}
Recall that our algorithm executes many DFS explorations in parallel, each originating at a free vertex. 
When a DFS exploration from a free vertex has no new edges to visit, we say that the free vertex becomes inactive; otherwise, it is active.
Our algorithm terminates when the number of active free vertices becomes a small fraction of the current matching size. The main goal of our correctness proof is to show that terminating the search for augmentations is justified. 
Speaking informally, the intuition behind our correctness argument is that if our algorithm runs indefinitely, each short augmentation will eventually intersect an augmentation our algorithm has already found.

Our proof of correctness boils down to showing the following:

\begin{mdframed}[backgroundcolor=gray!10, linecolor=brown!40!black, roundcorner=5pt]
    ({Informal version of \cref{lem:active}}) Consider a short augmenting path $P$ in the input graph $G$. Then, at any point in time, it holds: 
    \begin{itemize}
        \item our algorithm has already found an augmentation \emph{intersecting} $P$, or 
        \item there is an ongoing DFS exploration \emph{intersecting} $P$.
    \end{itemize}
\end{mdframed}

Recall that our algorithm maintains augmenting trees from free vertices. 
On a very high level, this enables us to think about augmentation search as, informally speaking, it is done on bipartite graphs. In particular, we show the following invariant:
\begin{mdframed}[backgroundcolor=gray!10, linecolor=brown!40!black, roundcorner=5pt]
    ({Informal version of \cref{inv:outer-independence}}) 
    At the beginning of every pass, if a non-tree edge induces an odd cycle, that edge cannot be reached by any so far discovered alternating path starting at a free vertex.
\end{mdframed}
One can also view this invariant as a way of saying that no relevant odd cycle is visible to the algorithm.
Of course, our algorithmic primitives and analysis must ensure that this view is indeed tree-like. Once we established this, we could bypass the technical intricacies of prior work.

\subsection{Pass and space complexity, and approximation guarantee (\cref{sec:pass-complexity})}
Our algorithm progressively finds a better approximation of the current approximate matching.
We implement that by dividing our augmentation search into different \emph{scales}.
A fixed scale guides the granularity of the search of the rest of the algorithm, and the scale values range from $1/2$ to $O(1/\eps^2)$ in powers of $2$. Each scale is further divided into many \emph{phases}. 

Our pass complexity balances and ties several parameters guiding the algorithm. These parameters are the number of edge-label reductions, the sizes of alternating trees, the scale values, and the number of phases in a scale. The most important of these parameters are scale and the upper bound on an alternating tree size.

When phases are executed for a given scale $h$, the attempt is to arrive at a $(1+O(h/\eps))$-approximate maximum matching. 
Hence, in the beginning, when there are many augmentations, large values of $h$ imply that the algorithm will soon arrive at the desired approximation. Importantly, this also means that fewer augmenting paths must be found for the next scale, i.e., scale $h / 2$, because scale $h/2$ starts with a better approximation than scale $h$. Therefore, scales enable us to balance the quality of approximation we want to achieve with the number of augmentations that must be found: the tighter the approximation requirement is, the slower the algorithm is; the fewer the augmentations must be found, the faster the algorithm is. 



\subsection{Comparison with \cite{FMU22}}
\label{sec:comparison-with-FMU}
A fundamental difference between our approach and that of \cite{FMU22} is that the search structure from a free vertex can be seen as a tree that we also refer to by structure. 
The same as in \cite{FMU22}, in our work, DFS structures $S_\alpha$ and $S_\beta$ originating at different free vertices $\alpha$ and $\beta$ might affect each other -- either by moving a part of $S_\alpha$ to $S_\beta$ via \algOvertake, or by finding an augmentation between $\alpha$ and $\beta$.
In \cite{FMU22}, these structures and blossoms are represented as a union of alternating paths, with some edges being marked as belonging to odd cycles. There is no special treatment of blossoms, nor do those structures have any particular shape.
On the other hand, we represent these structures as alternating trees; some vertices in those trees might correspond to blossoms. Crucially, it enables us to simplify structure-related procedures, provide a simpler proof of correctness, and prove new properties about structure sizes, allowing us to significantly reduce the pass complexity (more details are provided in \cref{sec:pass-complexity}). 

Finally, we believe the complexity of \cite{FMU22} can be reduced to $O(1/\eps^{16})$ by a slightly more careful analysis and tweaking parameters.
Adding the scales would improve the exponent in the pass-complexity by an additional $2$.
In addition, replacing ``a maximal set of augmenting paths'' with ``the maximum set of augmenting paths'', e.g., using our \cref{lem:short-path}, in the complexity analysis in \cite{FMU22} would result in yet another improvement by $2$ in the exponent of pass-complexity.
Nevertheless, it is unclear that, unless fundamental changes are made in the approach, \cite{FMU22} can result in a pass-complexity better than $O(1/\eps^{12})$.

\begin{remark}
The ideas of the truncated DFS exploration, maintaining edge labels, and the idea of the overtaking operation were first proposed in \cite{eggert2009bipartite} and later used in \cite{FMU22}.
The actual formulation of overtaking in our work is inspired by but different from \cite{eggert2009bipartite} or \cite{FMU22}.
\end{remark}

\section{Algorithms} \label{sec:alg}
This section presents our algorithm approach in detail. 
Its analysis is deferred to \cref{sec:correctness,sec:pass-complexity}.
We start by presenting two data structures that our algorithm maintains: the edge-exploration each free vertex maintains, which we call \emph{structure} (\cref{sec:structure}), and a label that each matched arc maintains (\cref{sec:label}).
In \cref{sec:algo-statement} we provide the base of our approach, which consists of many phases. The algorithms handling those phases are described in the subsequent sections, with \cref{sec:phase} providing an overview of a single phase.

\paragraph{Remark:} As already noted, throughout the paper, we attempt to use the algorithmic design as closely as possible to work prior, particularly the one used in \cite{FMU22}. 
We hope that in this way, we aid readers in comparing our contributions to priors.

\subsection{Algorithms' preliminaries}

\subsubsection{Free-vertex structures}
\label{sec:structure}


In our algorithm, each free vertex $\alp$ maintains a \emph{structure} (see \cref{fig:structure} for an example), defined as follows.

\begin{definition}[The structure of a free vertex] \label{def:structure}
The structure of a free vertex $\alp$, denoted by $S_\alp$, is a tuple $(G_\alp, \Omg_\alp, w'_\alp)$, where 
\begin{itemize}
    \item $G_\alp$ is a subgraph of $G$,
    \item $\Omg_\alp$ is a regular set of blossoms of $G_\alp$, and 
    \item $w'_\alp$ is either $\emptyset$ or an outer vertex of the alternating tree $G_\alp / \Omg_\alp$.
\end{itemize} 
Each structure $S_\alp$ satisfies the following properties.
\begin{enumerate}
    \item \textbf{Disjointness:} For any free vertex $\beta \neq \alp$, $G_\alp$ is vertex-disjoint from $G_\beta$.
    \item \textbf{Tree representation:} The subgraph $G_\alp$ contains a set of arcs satisfying the following: If $G_\alp$ contains an arc $(u, v)$ with $\Omg_\alp(u) \neq \Omg_\alp(v)$, then $\Omg_\alp(u)$ is the parent of $\Omg_\alp(v)$ in the alternating tree $G_\alp / \Omg_\alp$.
    \item \textbf{Unique arc property:} For each arc $(u', v') \in E(G_\alp / \Omg_\alp)$, there is a unique arc $(u, v) \in G_\alp$ such that $\Omg_\alp(u) = u'$ and $\Omg_\alp(v) = v'$. \label{def:unique-arc}
\end{enumerate}
\end{definition}

\noindent We denote the alternating tree $G_\alp / \Omg_\alp$ by $T'_\alp$.
Since $\Omg_\alp$ is a regular set of blossom, each inner vertex of $T'_\alp$ is a trivial blossom, whereas each outer vertex may be a non-trivial blossom.
\cref{fig:structure-b} shows $T'_\alp$ corresponding to the structure in \cref{fig:structure-a}.
We remark that $G_\alp$ may not be a vertex-induced subgraph.
That is, $G$ may contain arcs that are not in $G_\alp$ but connect two vertices in $G_\alp$.

\begin{definition}[The working vertex and active path of a structure] Consider a structure $S_\alp$. 
The \emph{working vertex} of $S_\alp$ is defined as the vertex $w'_\alp$, which can be $\emptyset$.
If $w'_\alp \neq \emptyset$, we define the \emph{active path} of $S_\alp$ as the unique path on $T'_\alp$ from the root $\Omg_\alp(\alp)$ to $w'_\alp$.
Otherwise, the active path is defined as $\emptyset$.
\end{definition}

\begin{definition}[Active vertices, arcs, and structures]
A vertex or arc of $T'_\alp$ is said to be \emph{active} if and only if it is on the active path. We say $S_\alp$ is active if $w'_\alp \neq \emptyset$.    
\end{definition}

\begin{figure}[h]
\centering
    \begin{subfigure}[t]{0.65\linewidth}\label{fig:structure-graph}
        \centering
        \includegraphics[width=\textwidth]{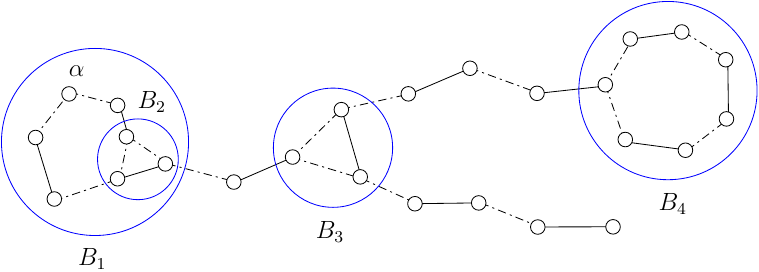}
        \caption{The graph $G_\alp$.}
        \label{fig:structure-a}
    \end{subfigure}
    \vskip\baselineskip
     \begin{subfigure}[t]{0.6\linewidth}
        \centering
        \includegraphics[width=\textwidth]{Sketches/structure_tree.pdf}
        \caption{The contracted graph $G_\alp/\Omg_\alp$.}
        \label{fig:structure-b}
    \end{subfigure}
    \caption{Example of a structure $S_\alp$. Dashed and solid edges denote the unmatched and matched edges, respectively.
    $\alp$ is a free vertex.
    $\Omg_\alp$ contains all trivial blossoms in $G_\alp$ and the non-trivial blossoms $\{B_1, B_2, B_3, B_4\}$, where $B_1 = \Omg_\alp(\alp)$ and $B_4$ is the working vertex $w'_\alp$ in $G_\alp/\Omg_\alp$. (a) The graph $G_\alp$. (b) The corresponding contracted graph $G_\alp/\Omg_\alp$. The encircled vertices correspond to the non-trivial blossoms. $w'_\alp$ is the working vertex and the highlighted path, from $B_1$ to $B_4$, is the active path.}
    \label{fig:structure}
\end{figure}

Let $F$ be the set of free vertices.
Throughout the execution, we maintain a set $\Omg$ of blossoms, which consists of all blossoms in $\bigcup_{\alp \in F} \Omg_\alp$ and all trivial blossoms.
Note that $\Omg$ is a laminar set of blossoms.
We denote by $G'$ the contracted graph $G / \Omg$.
The vertices of $G'$ are classified into three sets:
(1) the set of inner vertices, which contains all inner vertices in $\bigcup_{\alp \in F} V(T'_\alp)$;
(2) the set of outer vertices, which contains all outer vertices in $\bigcup_{\alp \in F} V(T'_\alp)$;
(3) the set of \emph{unvisited vertices}, which are the vertices not in any structure.

Similarly, we say a vertex in $G$ is \emph{unvisited} if it is not in any structure.
An arc $(u, v) \in G$ is a \emph{blossom arc} if $\Omg(u) = \Omg(v)$; otherwise, $(u, v)$ is a \emph{non-blossom arc}.
An \emph{unvisited arc} is an arc $(u, v) \in E(G)$ such that $u$ and $v$ are unvisited vertices.

\subsubsection{Labels of matched arcs} \label{sec:label}

Our algorithm stores the set of all matched arcs throughout its execution.
Each matched arc is associated with a \emph{label}, defined as follows.

\begin{definition}[The label of a matched arc]
Each matched arc $a^* \in G$ is assigned a label $\ell(a^*)$ such that $1 \leq \ell(a^*) \leq \lmax + 1$, where $\lmax$ is defined as $3 / \eps$.
\end{definition}

\noindent Each matched arc $a' \in G'$ corresponds to a unique non-blossom matched arc $a \in G$;
for ease of presentation, we denote by $\ell(a')$ the label of $a$.

Our algorithm maintains the following invariant.

\begin{invariant}[Increasing labeling] \label{inv:increasing-labeling} 
For any alternating path $(\Omg(\alp), a'_1, a'_2, \ldots, a'_k)$ on $T'_\alp$ starting from the root, it holds that $\ell(a'_1) < \ell(a'_2) < \dots < \ell(a'_k)$.
\end{invariant}


\subsection{Algorithm overview}
\label{sec:algo-statement}

In the following, we will sketch our algorithm, incrementally providing more details.
\cref{alg:outline} gives a high-level description of the algorithm. Recall that $\frac{1}{\eps}$ is assumed to be a power of $2$.

\begin{algorithm}
\begin{algorithmic}[1]
\medskip 
\Statex \textbf{Input:} a graph $G$ and the approximation parameter $\eps$
\Statex \textbf{Output:} a $(1+\eps)$-approximate maximum matching
\medskip
\Statex \hrule 

\State compute, in a single pass, a 2-approximate maximum matching $M$ \label{line:2-approx-matching}
\For{scale $h = \frac{1}{2}, \frac{1}{4}, \frac{1}{8}, \dots, \frac{\epsilon^2}{64}$\label{line:scale-h}}
    \For {phases $t = 1, 2, \dots, \frac{144}{h\eps}$ \label{line:call-phase-given-h}} 
        \State $\calP \gets \algPhase(G, M, \eps, h)$ \Comment{Nothing stored from the previous phase.}
        \State restore all vertices removed in the execution of \algPhase \label{line:restore}
        \State augment the current matching $M$ using the vertex-disjoint augmenting paths in $\calP$ \label{line:augment}
    \EndFor
\EndFor
\State \Return $M$
\end{algorithmic}
\caption{A high-level algorithm description.}
\label{alg:outline}
\end{algorithm}

\cref{alg:outline} provides an outline of our approach.
\cref*{line:2-approx-matching} applies a simple greedy algorithm to find a maximal matching, which is a 2-approximation for the problem.
Starting from this maximal matching, our algorithm repeatedly finds augmenting paths to improve the current matching.
This is done by executing several \emph{phases} with respect to different \emph{scales}, detailed as follows.

Each iteration of the for-loop in \cref*{line:scale-h} corresponds to a scale $h$.
In one scale $h$, each iteration of the for-loop in \cref*{line:call-phase-given-h} is called a phase with respect to the scale $h$.
In each phase, the procedure $\algPhase$ is invoked to find a set $\calP$ of vertex-disjoint augmenting paths.
In the execution of $\algPhase$, we may \textit{conceptually remove} some vertices from $G$.
After the execution of $\algPhase$, \cref*{line:restore} restores all removed vertices to $G$.
After this step, $G$ is identical to the input graph.
Then, \cref*{line:augment} augments the current matching using the set $\calP$ of vertex-disjoint augmenting paths, which increase the size of $M$ by $|\calP|$.

The scale $h$ is a parameter that determines the number of phases executed and the number of passes spent on each phase.
By passing a smaller scale to $\algPhase$, $\algPhase$ would spend more passes attempting to find more augmenting paths in the graph.
Our algorithm decreases the scale gradually so that more and more augmenting paths in the graph can be discovered.

\begin{algorithm}
\begin{algorithmic}[1]
\medskip 
\Statex \textbf{Input:} a graph $G$, the current matching $M$, the parameter $\eps$, and the current scale $h$
\Statex \textbf{Output:} a set $\calP$ of \emph{disjoint} $M$-augmenting paths
\medskip
\Statex \hrule 

\State $\calP \gets \emptyset$ \label{line:init-empty}
\State $\ell(a) \gets \lmax + 1$ for each arc $a \in M$ \label{line:init-label}
\State for each free vertex $\alp$, initialize its structure $S_\alp$ \label{line:init-structure}
\State compute parameters $\limit_h = \frac{6}{h} + 1$ and $\taumax(h) = \frac{72}{h\eps}$ \label{line:compute-par}
\For{$\bundle$s $\tau = 1, 2, \dots, \taumax(h)$} \label{line:for-pass}
    \For {each free vertex $\alp$} \label{line:init-for}
        \State if $S_\alp$ has at least $\limit_h$ vertices, mark $S_\alp$ as ``on hold'' \label{line:hold}
        \State if $S_\alp$ has less than $\limit_h$ vertices, mark $S_\alp$ as ``not on hold'' \label{line:not-hold}
        \State mark $S_\alp$ as ``not modified'' \label{line:not-modified}
    \EndFor
    \State $\algExtend$ (\cref{alg:extend}) \label{line:extend}
    \State $\algCheck$ \label{line:check}
    \State $\algBacktrack$ \label{line:backtrack}
\EndFor
\State \Return
\end{algorithmic}
\caption{$\algPhase$: the execution of a single phase.
}
\label{alg:phase}
\end{algorithm}

\subsection{A phase overview (\algPhase)} 
\label{sec:phase}
We now proceed to outline what the algorithm does in a single phase, whose pseudocode is given as \cref{alg:phase}.
In each phase, our algorithm executes DFS explorations from all free vertices in parallel.
Details of the parallel DFS are described as follows.
\Cref*{line:init-empty,line:init-label,line:init-structure} initialize the set of paths $\calP$, the label of each arc, and the structure of each free vertex.
The structure of a free vertex $\alp$ is initialized to be an alternating tree of a single vertex $\alp$.
That is, $G_\alp$ and $\Omg_\alp$ are set to be
    a graph with a single vertex $\alp$ and
    a set containing a single trivial blossom $\{\alp\}$, respectively;
the working vertex $w'_\alp$ is initialized as the root of $T'_\alp$, that is, $\Omg_\alp(\alp)$.
\Cref*{line:compute-par} computes two parameters $\limit_h$ and $\taumax(h)$.
The purpose of these two parameters is detailed later.
The for-loop in \Cref*{line:for-pass} executes $\taumax(h)$ iterations, where each iteration is referred to as a \emph{\bundle}.
The execution of a $\bundle$ corresponds to one step in the parallel DFS.
Each $\bundle$ consists of four parts:
\begin{enumerate}[(1)]
    \item \Cref*{line:init-for,line:hold,line:not-hold,line:not-modified} initialize the status of each structure in this $\bundle$. A structure is marked as \emph{on hold} if and only if it contains at least $\limit_h$ vertices. Each structure $S_\alp$ is marked as \emph{not modified}. The purpose of this part is described in \cref{sec:marking}.
    \item The procedure $\algExtend$ makes a pass over the stream and attempts to extend each structure that is not on hold. Details of this procedure are given in \cref{sec:extend}.
    \item After the execution of $\algExtend$, the subgraph $G_\alp$ maintained in each structure $S_\alp$ may change. 
    The procedure $\algCheck$ is then invoked to identify blossoms and augmenting paths. 
    The procedure makes a pass over the stream, contracts some blossoms that contain the working vertex of a structure, and identifies pairs of structures that can be connected to form augmenting paths. Details of this procedure are given in \cref{sec:check}.
    \item The procedure $\algBacktrack$ examines each structure. If a structure is not on hold and fails to extend in this pass, $\algBacktrack$ backtracks the structure by removing one matched arc from its active path. Details of this procedure are given in \cref{sec:backtrack}.
\end{enumerate}

In \cref{sec:appendix-inv}, we prove the following lemma, showing that all invariants presented in \cref{sec:structure,sec:label} are preserved in the execution of $\algPhase$.

\begin{restatable}{lemma}{invpreservation}
\label{lem:inv-preservation} The following holds throughout the execution of $\algPhase$. For each free vertex $\alp$ that is not removed, $S_\alp$ is a structure per \cref{def:structure}; in addition, \cref{inv:increasing-labeling} holds.
\end{restatable}

\subsection{Marking a structure on hold or modified} \label{sec:marking}
In the for-loop of \cref*{line:init-for}, we mark a structure $S_\alp$ \textit{on hold} if and only if it contains at least $\limit_h$ vertices.
See \cref*{line:hold,line:not-hold} of \cref{alg:phase}.
This operation plays a crucial role in our analysis of the pass-complexity of our algorithm, e.g., \cref{lem:active-bound,lem:structure-bound}.

In the for-loop, we also mark each structure as \textit{not modified}.
Recall that each structure $S_\alp$ is represented by a tuple $(G_\alp, \Omg_\alp, w'_\alp)$.
In the execution of the $\bundle$, we may modify some structures by, for example, adding new arcs to $G_\alp$.
Whenever $G_\alp, \Omg_\alp$, or $w'_\alp$ is changed, we mark $S_\alp$ as modified.
In other words, if a structure $S_\alp$ is marked as not modified, $(G_\alp, \Omg_\alp, w'_\alp)$ is unchanged since the beginning of the current $\bundle$.

\subsection{Basic operations on structures} \label{sec:basic-operations}
In the following, we present three basic operations for modifying the structures. These operations are used in the execution of $\algExtend$ and $\algCheck$. Whenever one of these operations is applied, the structures involved are marked as modified.

\subsubsection{Procedure $\algAugment(g, \calP)$}
\label{sec:augment}
\begin{minipage}{0.95\linewidth}
    \begin{mdframed}[backgroundcolor=gray!15, linecolor=red!40!black]
    \textbf{Invocation reason:} When our algorithm discovers an augmenting path in $G$.
    
    \textbf{Input:}
    
    \mytab - The set $\calP$.
    
    \mytab - An unmatched arc $g = (u, v)$, where $g \in E(G)$. The arc $g$ must satisfy the following property: $\Omg(u)$ and $\Omg(v)$ are outer vertices of two different structures.
    \end{mdframed}
\end{minipage}
\\\\
Since $\Omg(u)$ is an outer vertex, $T'_\alp$ contains an even-length alternating path from the root $\Omg(\alp)$ to $\Omg(u)$.
Similarly, $T'_\beta$ contains an even-length alternating path from $\Omg(\beta)$ to $\Omg(v)$.
Since there is an unmatched arc $(\Omg(u), \Omg(v))$ in $G'$, the two paths can be concatenated to form an augmenting path $P'$ on $G'$.

By using \cref{lem:even_path}, we obtain an augmenting path $P$ on $G$ by replacing each blossom on $P'$ with an even-length alternating path.
$\algAugment$ adds $P$ to $\calP$ and removes $S_\alp$ and $S_\beta$.
That is, all vertices from $V(G_\alp) \cup V(G_\beta)$ are removed from $G$, and $\Omg$ is updated as $\Omg - (\Omg_\alp \cup \Omg_\beta)$.
The vertices remain removed until the end of $\algPhase$.
This guarantees that the paths in $\calP$ remain disjoint.
Recall that our algorithm adds these vertices back before the end of this phase, when \cref*{line:restore} of \cref{alg:outline} is executed.

\subsubsection{Procedure $\algContract(g)$}
\label{sec:contract}
\begin{minipage}{0.95\linewidth}
    \begin{mdframed}[backgroundcolor=gray!15, linecolor=red!40!black]
    \textbf{Invocation reason:} When a blossom in a structure is discovered.
    
    \textbf{Input:} 

    \mytab - An unmatched arc $g = (u, v)$, where $g \in E(G)$, such that $\Omg(u)$ and $\Omg(v)$ are distinct outer vertices in the same structure, denoted by $S_\alp$. In addition, $\Omg(u)$ is the working vertex of $S_\alp$

    \end{mdframed}
\end{minipage}
\\\\
Let $g'$ denote the arc $(\Omg(u), \Omg(v))$.
By \cref{lem:contraction}, $T'_\alp \cup \{g'\}$ contains a unique blossom $B$.
The procedure contracts $B$ by adding $B$ to $\Omg_\alp$; hence, $T'_\alp$ is updated as $T'_\alp / B$ after this operation.
The arc $g$ is added to $G_\alp$.

By \cref{lem:contraction}, $T'_\alp$ remains an alternating tree after the contraction, and $B$ becomes an outer vertex of $T'_\alp$.
Next, the procedure sets the label of each matched arc in $E(B)$ to $0$.
(After this step, for each matched arc $a \in E(B)$, both $\ell(a)$ and $\ell(\cev{a})$ are $0$.)

Note that the working vertex of $S_\alp$, that is, $\Omg(u)$, is contracted into the blossom $B$.
The procedure then sets $B$ as the new working vertex of $S_\alp$.
Then, $S_\alp$ is marked as modified.


\subsubsection{Procedure $\algOvertake(g, a, k)$}
\label{sec:overtake}
\begin{minipage}{0.95\linewidth}
    \begin{mdframed}[backgroundcolor=gray!15, linecolor=red!40!black]
    \textbf{Invocation reason:} When the active path of a structure $S_\alp$ can be extended through $g$ to \emph{overtake} the matched arc $a$ and reduce $\ell(a)$ to $k$.
    
    \textbf{Input:}

    \mytab - An unmatched arc $g = (u, v) \in G$.

    \mytab - A non-blossom matched arc $a = (v, t) \in G$, which shares the endpoint $v$ with $g$.
    
    \mytab - A positive integer $k$.

    \mytab - The input must satisfy the following.
        \begin{itemize}
            \item[] (P1) $\Omg(u)$ is the working vertex of a structure, denoted by $S_\alp$.
            \item[] (P2) $\Omg(v) \neq \Omg(u)$, and $\Omg(v)$ is either an unvisited vertex or an inner vertex of a structure $S_\beta$, where $S_\beta$ can be $S_\alp$. In the case where $\Omg(v) \in S_\alp$, $\Omg(v)$ is not an ancestor of $\Omg(u)$.
            \item[] (P3) $k < \ell(a)$.
        \end{itemize}

    \end{mdframed}
\end{minipage}
\\\\
For ease of notation, we denote $\Omg(u), \Omg(v)$, and $\Omg(t)$ by $u', v',$ and $t'$, respectively.
Since $v'$ is not an outer vertex, it is the trivial blossom $\{v\}$.
The procedure $\algOvertake$ performs a series of operations, detailed as follows.
Consider three cases, where in all of them we reduce the label of $a$ to $k$.

\vspace{2mm}
\noindent \textbf{Case 1.} $a$ is not in any structure.
We include the arcs $g$ and $a$ to $G_\alp$.
The trivial blossoms $v'$ and $t'$ are added to $\Omg_\alp$.
The working vertex of $S_\alp$ is updated as $t'$, which is an outer vertex of $T'_\alp$.
Then, $S_\alp$ is marked as modified.

\vspace{3mm}
\noindent \textbf{Case 2.} $a$ is in a structure $S_\beta$.
By the definition of $g$, $v'$ is an inner vertex.
Thus, $v'$ is not the root of $T'_\beta$.
By \cref{def:structure}-\ref{def:unique-arc}, there is a unique unmatched arc $(p, v) \in S_\beta$ such that $\Omg(p)$ is the parent of $\Omg(v)$.
Two subcases are considered.
\begin{figure}
    \centering
    \includegraphics[width=0.5\textwidth]{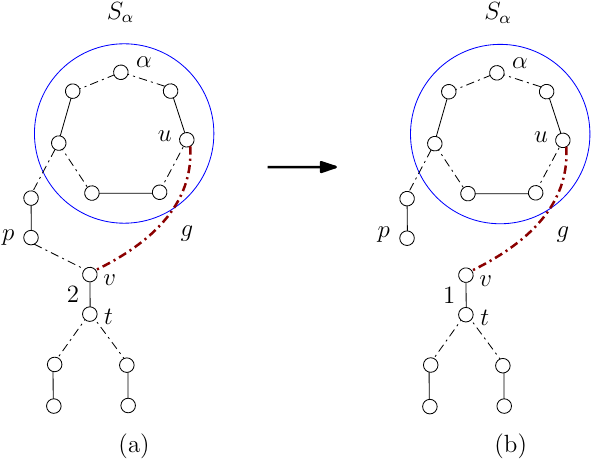}
    \caption{Example of Case 2.1 of procedure $\algOvertake$. In this case, $\Omg(u)$ and $\Omg(v)$ are in the same structure. (a) The structure $S_\alp$ before, where the arc $g$ is highlighted. (b) The structure $S_\alp$ after we invoke $\algOvertake$ via $g = (u,v)$.}
    \label{fig:overtaking-case2-1}
\end{figure}

\begin{itemize}
    \item[] \textbf{Case 2.1.} $\alp = \beta$. See \cref{fig:overtaking-case2-1} for an example. By (P2), $v'$ is not an ancestor of $u'$.
    The overtaking operation is done by updating $E(G_\alp)$ as $E(G_\alp) - \{(p, v)\} \cup \{g\}$.
    On the tree $T'_\alp$, this operation corresponds to re-assigning the parent of $v'$ as $u'$.
    Then, we update the working vertex of $S_\alp$ as $t'$ and mark $S_\alp$ as modified.

\item[] \textbf{Case 2.2.} $\alp \neq \beta$.
    See \cref{fig:overtaking-case2-2} for an example.
    Similar to Case 2.1, the objective of the overtaking operation is to re-assign the parent of $v'$ as $u'$ on $G'$.
    However, we need to handle several additional technical details in this case.
    The overtaking operation consists of the following steps.
    \begin{itemize}
        \item[] Step 1: Remove the arc $(p, v)$ from $G_\beta$ and add the arc $(u, v)$ to $G_\alp$.
        \item[] Step 2: Move, from $G_\beta$ to $G_\alp$, all vertices $x$ such that $\Omg(x)$ is in the subtree of $v'$
        \item[] Step 3: Move, from $G_\beta$ to $G_\alp$, all arcs $(x, y)$ where $x$ and $y$ are both moved in Step 2.
        \item[] Step 4: Move, from $\Omg_\beta$ to $\Omg_\alp$, all blossoms that contain a subset of vertices moved in Step 2.
        \item[] Step 5: If the working vertex of $S_\beta$ was under the subtree of $t'$ before Step 1, we set $w'_\alp$ as $w'_\beta$ and then update $w'_\beta$ as $\Omg(p)$. Otherwise, set $w'_\alp$ as $t'$.
    \end{itemize}
    After the overtaking operation, both $S_\alp$ and $S_\beta$ are marked as modified.
\end{itemize}

\begin{figure}
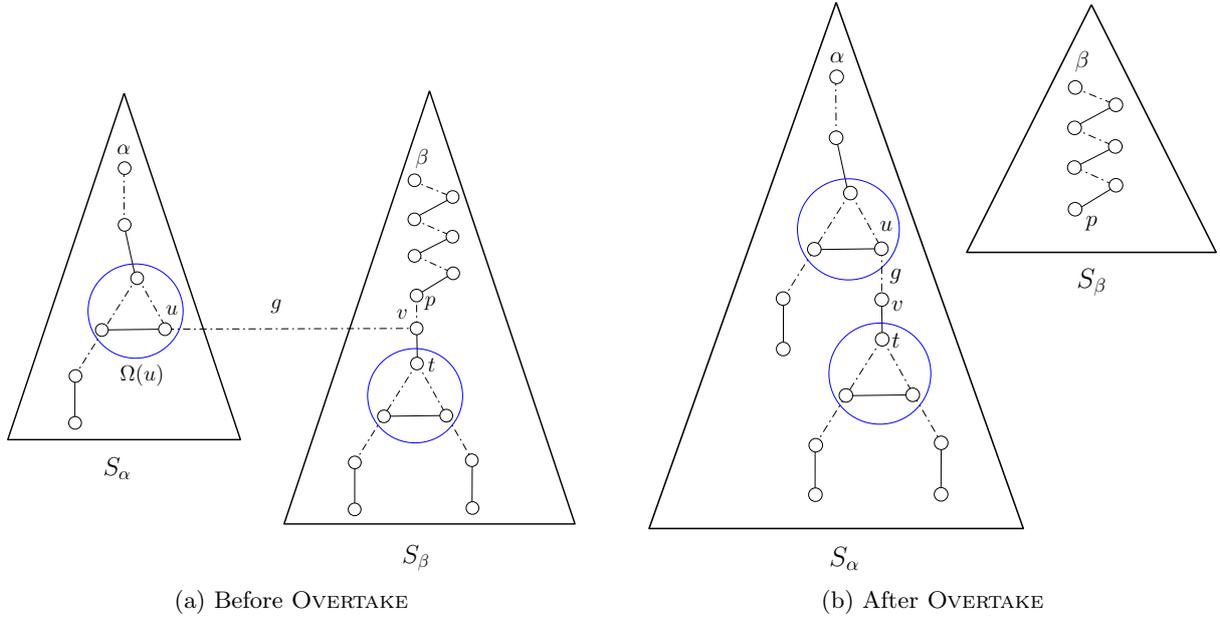

\centering
    \begin{subfigure}[t]{0.47\linewidth}
        \centering
        \includegraphics[width=\textwidth]{Sketches/overtaking-1.pdf}
        \caption{Before $\algOvertake$}
    \end{subfigure}
    \hfill
     \begin{subfigure}[t]{0.47\linewidth}
        \centering
        \includegraphics[width=\textwidth]{Sketches/overtaking-2.pdf}
        \caption{After $\algOvertake$}
    \end{subfigure}
    \caption{Example of Case 2.2 of the procedure $\algOvertake$ where $g = (u, v)$ connects the two structures $S_\alp$ and $S_\beta$. While, in this example $\Omg(p) = \{p\}$ is a trivial blossom, it can be a non-trivial blossom in general.}
    \label{fig:overtaking-case2-2}
\end{figure}

\subsection{Procedure \algExtend} \label{sec:extend}
The goal of $\algExtend$ is to \textit{extend} each structure $S_\alp$, where $S_\alp$ is not on hold, by performing at most one of the $\algAugment$, $\algContract$, or $\algOvertake$ operations.

\begin{algorithm}
\begin{algorithmic}[1]
\medskip 
\Statex \textbf{Input:} a graph $G$, the parameter $\eps$, the current matching $M$, the structure $S_\alp$ of each free vertex $\alp$, the set of paths $\calP$
\medskip 
\Statex \hrule 

\For{each arc $g = (u, v) \in E(G)$ on the stream}
    \If{$u$ or $v$ was removed in this phase}
        \State continue with the next arc
    \EndIf
    
    \If{$\Omg(u) = \Omg(v)$, or $\Omg(u)$ is not the working vertex of any structure, or $g$ is matched}
        \State continue with the next arc
    \EndIf
    \If{$u$ belongs to a structure that is marked as modified or on hold}
        \State continue with the next arc
    \EndIf
    \If {$\Omg(v)$ is an outer vertex}
        \If {$\Omg(u)$ and $\Omg(v)$ are in the same structure}
            \State $\algContract(g)$
        \Else
            \State $\algAugment(g)$
        \EndIf
    \Else \Comment{$\Omg(v)$ is either unvisited or an inner vertex.}
        \State compute $\lab(u)$ \label{line:shortest-path}
        \State $a \gets$ the matched arc in $G$ whose tail is $v$
        \If {$\lab(u) + 1 < \ell(a)$}
            \State $\algOvertake(g, a, \lab(u) + 1)$
        \EndIf
    \EndIf    
\EndFor
\end{algorithmic}
\caption{The execution of $\algExtend$.} \label{alg:extend}
\end{algorithm}

The procedure works as follows.
(See \cref{alg:extend} for a pseudocode.)
The algorithm makes a pass over the stream to read each arc $g = (u, v)$ of $G$.
When an arc $g$ is read, it is mapped to an arc $g' = (\Omg(u), \Omg(v))$ of $G'$.
In $\algExtend$, we only consider non-blossom unmatched arcs whose tail is a working vertex.
Hence, if $\Omg(u) = \Omg(v)$, $\Omg(u)$ is not the working vertex of a structure, or $g$ is a matched arc, then we simply ignore $g$.
If one of $u$ or $v$ is removed, we also ignore $g$.

Let $S_\alp$ denote the structure whose working vertex is $\Omg(u)$, and let $S_\beta$ denote the structure containing $\Omg(v)$.
We ignore $g$ if $S_\alp$ is marked as on hold.
Before this procedure, no structure is marked as modified.
If $\alp$ is marked as modified, we know it has extended during the execution of $\algExtend$.
In this case, we also ignore $g$.
This ensures that each structure only extends once in the execution of $\algExtend$.
The above property is crucial to our analysis, e.g., in \cref{lem:active}.

We examine whether $g$ can be used for extending $S_\alp$ as follows.
\begin{itemize}
    \item[] \textbf{Case 1:} $\Omg(v)$ is an outer vertex and $S_\alp = S_\beta$. In this case, $g'$ induces a blossom on $T'_\alp$. We invoke $\algContract$ on $g$ to contract this blossom.
    \item[] \textbf{Case 2:} $\Omg(v)$ is an outer vertex and $S_\alp \neq S_\beta$. In this case, the two structures can be connected to form an augmenting path. We invoke $\algAugment$ to compute this augmenting path and remove the two structures.
    \item[] \textbf{Case 3:} $\Omg(v)$ is either an inner vertex or an unvisited vertex.
    Note that $v$ cannot be a free vertex because, for each free vertex $\gamma$, it holds that $\Omg(\gamma)$ is an outer vertex.
    Therefore, $v$ is the tail of a matched arc $a$.
    We determine whether $S_\alp$ can overtake $a$ by computing a number $\lab(u) + 1$ and compare it with $\ell(a)$.
    The number $\lab(u)$ represents the last label in the active path, which is computed as follows:
    If $\Omg(u)$ is a free vertex, $\lab(u)$ is set to $0$;
    otherwise, $\lab(u)$ is the label of the matched arc in $G'$ whose head is $\Omg(u)$.
    If $\lab(u) + 1 < \ell(a)$, $\algOvertake$ is invoked to update the label of $a$ as $\lab(u) + 1$.
    
    We remark that $\algOvertake$ is never invoked when $\Omg(v)$ is an ancestor of $\Omg(u)$ in $T'_\alp$, because the sequence of labels in any root-to-leaf path is increasing (due to \cref{inv:increasing-labeling}).

\end{itemize}

\subsection{Procedure \algCheck} \label{sec:check}
The procedure $\algCheck$ performs two steps to identify augmenting paths and blossoms:

\begin{itemize}
    \item[] Step 1: Repeatedly invoke $\algContract$ on an arc connecting two outer vertices of the same structure, where one of the outer vertices is the working vertex.
    \item[] Step 2: For each arc $g$ connecting outer vertices of different structures, invoke $\algAugment$ with $g$.
\end{itemize}

Step 1 is implemented as follows.
First, for each free vertex $\alp$, compute $A_\alp$ as the set of arcs connecting two vertices in $G_\alp$.
The computation of $A_\alp$ for all free vertices $\alp$ can be done in one pass over the stream.
Next, we repeatedly perform the following operation on each structure $S_\alp$:
While there exists an arc $(u, v) \in A_\alp$ such that $\Omg(u)$ is the working vertex of $S_\alp$ and $\Omg(v) \neq \Omg(u)$ is an outer vertex, invoke $\algContract$ on $(u, v)$ to contract the blossom.

Step 2 is implemented by scanning each arc $(u, v)$ over the stream, and if $\Omg(u)$ and $\Omg(v)$ are outer vertices of different structures, we invoke $\algAugment$ on $(u, v)$.

The purpose of $\algCheck$ is to ensure an invariant (see \cref{inv:outer-independence}) that we leverage in our correctness analysis.

\subsection{Procedure \algBacktrack} \label{sec:backtrack}
The purpose of $\algBacktrack$ is to backtrack the structures that do not make progress in a $\bundle$. More formally, each structure that is not on hold or modified is backtracked. 
Note that if a structure is marked as not modified when the procedure $\algBacktrack$ is invoked, then it did not extend in $\algExtend$, and $\algCheck$ could not contract any blossom inside the structure.

Consider a structure $S_\alp$ that is not on hold and not modified. The backtracking is performed as follows.
If $w'_\alp$ is a non-root outer vertex of $T'_\alp$, we update the working vertex as the parent of the parent of $w'_\alp$, which is an outer vertex.
Otherwise, $w'_\alp$ is the root $\Omg(\alp)$, and we set the working vertex of $S_\alp$ as $\emptyset$, which makes $S_\alp$ inactive.

$\algBacktrack$ may change $w'_\alp$ of some structures $S_\alp$. We do not mark these changed structures as modified because the mark is not used in the rest of the $\bundle$.






\section{Correctness}
\label{sec:correctness}
The goal of this section is to show that our algorithm does not miss any short augmentation. That is, if our algorithm is left to run indefinitely and no structure is on hold, then at some point, the remaining graph will have no short augmentation left. 
Formally, 
we show the following.

\begin{definition}[Critical arc and vertex]
\label{def:critical}
Recall that the active path of a structure is a path in $G'$.
We say a non-blossom arc $(u, v) \in G$ is \emph{critical} if the arc $(\Omg(u), \Omg(v)) \in G'$ is active.
In particular, all blossom arcs in a structure are not critical, even if they are in an active blossom.
We say a free vertex $\alp \in G$ is \emph{critical} if $S_\alp$ is active.
\end{definition}

\begin{theorem}[No short augmenting paths is missed] \label{lem:active}
At the beginning of each $\bundle$, the following holds.
Let $P = (\alp, a_1, a_2, \dots, a_k, \beta)$ be an augmenting path in $G$ such that no vertex in $P$ is removed in this phase and $k \leq \lmax$.
At least one of the following holds: $\alp$ is critical, or $P$ contains a critical arc. 
\end{theorem}

\paragraph{Properties of a phase.}

Our analysis of \cref{lem:active} relies on the following two simple properties.

\begin{lemma}[Outer vertex has been a working one]
\label{lem:working} 
Consider a $\bundle$ $\tau$. Suppose that $G'$ contains an outer vertex $v'$ at the beginning of $\tau$. Then, there exists a $\bundle$ $\tau' \leq \tau$ such that $v'$ is the working vertex at the beginning of $\tau'$.
\end{lemma}

\begin{restatable}{invariant}{invouter}
\label{inv:outer-independence} At the beginning of each $\bundle$, no arc in $G'$ connects two outer vertices.    
\end{restatable}

\begin{lemma}
\label{lem:outer-independence} 
\cref{inv:outer-independence} holds.
\end{lemma}
Proofs of these two claims are deferred to \cref{sec:lem:working,sec:lem:outer-independence}.

We remark that \cref{inv:outer-independence} only holds at the beginning of each $\bundle$. During the execution of $\algExtend$, some structures may include new unvisited vertices or contract a blossom in the structure. These operations create new outer vertices that may be adjacent to existing outer vertices.

\subsection{No short augmentation is missed (Proof of \cref{lem:active})}

For a $\bundle\ \tau$, we use $\Omg^\tau$ and $\ell^\tau$ to denote, respectively, the set of blossoms and labels at the beginning of $\tau$.

Consider a fixed phase and a $\bundle$ $\tau$ in the phase.
Suppose, toward a contradiction, that at the beginning of $\tau$, there exists an augmenting path $P = (\alp, a_1, a_2, \dots, a_k, \beta)$ in $G$, where $k \leq \lmax$, such that:
\begin{enumerate}[(i)]
    \item none of the vertices in $P$ is removed in the phase, and
    \item $\alp$ and all arcs in $P$ are not critical. 
\end{enumerate}
Recall that for an arc to be critical, by \cref{def:critical}, it has to be non-blossom. Hence, some blossom arcs of $P$ may be inside of an active blossom at this moment.
For $i = 1, 2, \dots, k$, let $u_i$ and $v_i$ be the tail and head of $a_i$, respectively; that is, $a_i = (u_i, v_i)$.
Let $v_0 = \alp$.
Two cases are considered:

    
\paragraph{Case 1: There exists an index $q$ such that $\ell^{\tau}(a_q) > q$.}   

    Let $q$ be the smallest index such that $\ell^{\tau}(a_q) > q$.
    Since $\ell^{\tau}(a_q) > q > 0$, $a_q$ must be a non-blossom arc.
    Let $p$ be the smallest index such that $p < q$ and $a_{p+1}, \dots, a_{q-1}$ are blossom arcs.
    See \Cref{fig:short-augment} for a sketch example.

    \begin{figure}[h]
    \centering
    \includegraphics[width=.7\textwidth]{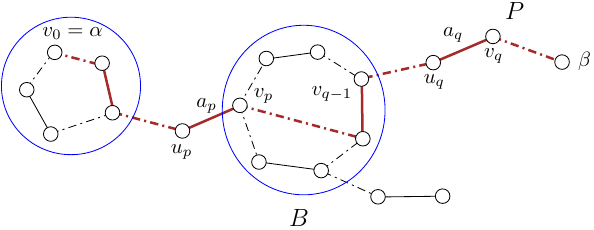}
    \caption{
    This sketch is used to illustrate the proof of \cref{lem:active}, Case~1.
    The highlighted path $P$ is an augmenting path between $\alp$ and $\beta$.
    The blossom $B$ contains all vertices and arcs in the path from $v_p$ to $a_{q-1}$.
    Note that $P$ contains an arc that is not in $E_B$.
    }
    \label{fig:short-augment}
    \end{figure}
    
    We first show that $\Omg(v_p)$ is an outer vertex, and all vertices in the path $(v_p, a_{p+1}, a_{p+2}, \dots, a_{q-1})$ are in the same inactive blossom.
    If $p = 0$, then $v_p = \alp$ and thus $\Omg(v_p)$ is an outer vertex.
    Otherwise, since $\ell^{\tau}(a_p) \leq p \leq \lmax$, $a_p$ is visited.
    Hence, $a_p$ is a non-blossom arc contained in a structure, which also implies that $\Omg(v_p)$ is an outer vertex.
    For $p < i < q$, since $a_i$ is a blossom arc, $\Omg(u_i)$ and $\Omg(v_i)$ are both outer vertices.
    Hence, by \cref{inv:outer-independence}, all vertices in the path $(v_p, a_{p+1}, a_{p+2}, \dots, a_{q-1})$ must be in the same blossom at the beginning of $\tau$.
    Denote this blossom by $B$.
    Clearly, if $p = 0$, then $v_p = \alp$ is the base of $B$.
    Otherwise, since $a_p$ is a non-blossom matched arc adjacent to $B$, we also have $v_p$ as the base of $B$.
    Since $a_p$ is non-critical at the beginning of $\tau$, we have that \textbf{$B$ is inactive at the beginning of $\tau$}.

    Let $\tau' \leq \tau$ be the last $\bundle$ such that $B$ is the working vertex of some structure $S_\gamma$ at the beginning of $\tau'$.
    By \cref{lem:working}, $\tau'$ exists, and since $B$ is inactive at the beginning of $\tau$, we further know that $\tau'$ < $\tau$.
    In $\tau'$, $S_\gamma$ backtracks from $B$, and $B$ remains inactive until at least the beginning of $\tau$.
    Hence, if $p > 0$, $\ell(a_p)$ is not updated between the end of $\tau'$ and the beginning of $\tau$.
    Therefore, \textbf{$\ell^{\tau'}(a_p) = \ell^{\tau}(a_p) \leq p$}.

    By definition of $\algBacktrack$ (\cref{sec:backtrack}), $S_\gamma$ is not marked as on hold or modified in $\tau'$.
    Thus, when $\algExtend$ read the unmatched arc $g = (v_{q-1}, u_q) \in G$ from the stream, $\Omg(v_{q-1}) = B$ is the working vertex of $S_\gamma$.
    We claim that $\cev{a_q}$ cannot be in any structure at this moment.
    If $\cev{a_q}$ is in any structure, then $\Omg(u_q)$ is an outer vertex, and $\algExtend$ should invoke either $\algAugment$ or $\algContract$ on $g$.
    This leads to a contradiction because either $S_\gamma$ is removed or it is marked as modified in this $\bundle$.
    Hence, $\cev{a_q}$ is not in any structure, and thus $\Omg(u_q)$ is either unvisited or an inner vertex.
    That is, \textbf{at the moment $g$ is read, \cref*{line:shortest-path} of $\algExtend$ is executed to compute $\lab(v_{q-1})$}.

    By the definition of $\lab(v_{q-1})$ (in \cref{sec:extend}), if $\Omg(v_{q-1}) = B$ is a free vertex, then $\lab(v_{q-1})$ is $0$;
    otherwise, $\lab(v_{q-1})$ is the label of the matched arc adjacent to $\Omg(v_{q-1})$, which is $\ell^{\tau'}(a_p)$.
    In either case, we have $\lab(v_{q-1}) \leq p$.
    This implies that $\algExtend$ should have updated the label of $a_q$ to $\lab(v_{q-1}) + 1 \leq p + 1 < q$ in $\tau'$, which contradicts the definition of $q$.

        
    
    \paragraph{Case 2: For each $i = 1, 2, \ldots, k$ it holds $\ell^{\tau}(a_i) \leq i$.}
    Let $p \leq k$ be the smallest index such that all vertices in $a_{p+1}, a_{p+2}, \dots, a_k$ are blossom arcs at the beginning of $\tau$. 
    Similar to Case 1, if $p = 0$, then $v_p = \alp$; otherwise, $a_p$ is a non-blossom arc with $\ell^{\tau}(a_p) \leq p \leq \lmax$. 
    Hence, $\Omg(v_p)$ is an outer vertex.

    For each $i > p$, since $a_i$ is a blossom arc, $\Omg(u_i)$ and $\Omg(v_i)$ are both outer vertices.
    Hence, by \cref{inv:outer-independence}, all vertices in the path $(v_p, a_{p+1}, a_{p+2}, \dots, a_{k})$ must be in the same blossom at the beginning of $\tau$.
    Denote this blossom by $B$.
    Since $\Omg^\tau(\beta)$ is the root of $T'_\beta$, it is also an outer vertex.
    That is, $\Omg^\tau(v_k)$ and $\Omg^\tau(\beta)$ are both outer vertices.
    
    Since $P$ contains an unmatched arc $(v_k, \beta)$, by \cref{inv:outer-independence}, $\Omg^\tau(\beta) = \Omg^\tau(v_k) = B$.
    If $p = 0$, this leads to a contradiction because $B$ contains two free vertices $v_p = \alp$ and $\beta$. If $p > 0$, then $B = \Omg(\beta)$ is not a free vertex because it is adjacent to a non-blossom matched arc $a_p$, which is also a contradiction. 

\subsection{Outer vertex has been a working one (Proof of \cref{lem:working})}
\label{sec:lem:working}
Suppose that $v'$ is a trivial blossom.
Then, $v'$ is a working vertex when it is first added to a structure.
Denote by $\tau'' < \tau$ the $\bundle$ in which $v'$ is first added to a structure.
We claim that $v'$ is the working vertex at the beginning of $\tau'' + 1$.
After $v'$ is added to a structure in $\tau''$, $v'$ may be overtaken by other structures in the same $\bundle$.
By the definition of $\algOvertake$ (\cref{sec:overtake}), when a structure overtakes $v'$, it also sets $v'$ to be its working vertex.
Hence, $v'$ is a working vertex after the invocation of $\algExtend$ (\cref{sec:extend}) in the $\bundle$ $\tau''$.
Since $v'$ remains in $G'$ at least until $\tau$, we know that $v'$ is not removed or contracted by $\algAugment$ (\cref{sec:augment}) or $\algContract$ (\cref{sec:contract}).
Therefore, $v'$ is the working vertex of a structure at the beginning of $\tau'' + 1$.

Next, consider the case when $v'$ is a non-trivial blossom. Then, $v'$ is a working vertex at the moment it is added to $\Omg$ by $\algContract$. Let $\tau'' < \tau$ denote the $\bundle$ such that $v'$ is added to $\Omg$. By a similar argument, we know that $v'$ is the working vertex of a structure at the beginning of $\tau'' + 1$.

\subsection{No arc between outer vertices (Proof of \cref{lem:outer-independence})}
\label{sec:lem:outer-independence}
Suppose, by contradiction, that at the beginning of some $\bundle$ $\tau$, there is an arc $g' \in E(G')$ connecting two outer vertices $u'$ and $v'$. Let $\tau_{u'} < \tau$ (resp. $\tau_{v'} < \tau$) be the first $\bundle$ in which $u'$ (resp. $v'$) is added to a structure by either $\algOvertake$ or $\algContract$. Without loss of generality, assume that $\tau_{u'} \geq \tau_{v'}$. There are two cases:
\begin{enumerate}
    \item $u'$ was an unvisited vertex before $\tau_{u'}$ and is added to a structure by $\algOvertake$ in $\tau_{u'}$.
    \item $u'$ is a non-trivial blossom added to $\Omg$ by $\algContract$ in $\tau_{u'}$. 
\end{enumerate}

In the first case, by the proof of \cref{lem:working}, we know that $u'$ is a working vertex after the invocation of $\algExtend$ in $\tau_{u'}$. Therefore, in $\algCheck$ (\cref{sec:check}), either $\algAugment$ or $\algContract$ will be invoked on $g'$, which contradicts the definition of $g'$. 

In the second case, $u'$ remains a working vertex after it is added to $\Omg$ (by $\algContract$, in either $\algExtend$ or $\algCheck$) and before the end of $\tau_{u'}$. Hence, in $\algCheck$, either $\algAugment$ or $\algContract$ will be invoked on $g'$, leading to a contradiction.

\section{Pass and Space Complexity, and Approximation Guarantee}
\label{sec:pass-complexity}
In this section, we show that our algorithm finds a $(1 + \eps)$-approximate matching using $O(n/\eps^6)$ words of space and $O(1/\eps^6)$ passes over the stream.

\subsection{Overview of parameters} \label{sec:overview-analysis}
As a reminder, our algorithm is parameterized by $h$ (see \cref{alg:outline}), which we refer to by scale. Intuitively, $h$ dictates the fraction of remaining augmentations our algorithm seeks. The idea is that initially, our algorithm is aggressive and looks for many augmentations at once. 
Only after, once the algorithm knows there are not many augmentations left, does the algorithm fine-tune by searching for a relatively few remaining augmentations. 
Since not many augmentations are left, this fine-tuning can be implemented more efficiently than without varying scales.

Our analysis ties together several parameters used in our algorithm, summarized in \cref{tab:parameter}.
Those parameters are:
\begin{itemize}
    \item $\hmin$, which is the minimum scale,
    \item $\tmax(h)$, which is the number of phases for a given scale,
    \item $\Delta_h$, which is an upper bound on the number of vertices that a structure can have.
\end{itemize}

\begin{table}[h]
    \centering
    \setlength{\extrarowheight}{4pt}
    \begin{tabular}{|c|l|c|c|}
    \hline
    Parameter & Definition  & Asymptotic value & Exact value \\ \hline
    $\lmax$ & maximum label for a visited matched arc  & $\Theta(\eps^{-1})$ & $3\eps^{-1}$ \\ \hline
    $\hmin$      & the minimum scale  & $\Theta(\eps^2)$ & $\eps^2 / 64$ \\ \hline
    $\tmax(h)$   & number of phases in a scale  & $\Theta((h\eps)^{-1})$ & $144(h\eps)^{-1}$ \\ \hline
    $\taumax(h)$ & number of $\bundle$s in a phase  & $\Theta((h\eps)^{-1})$ & $72(h\eps)^{-1}$ \\ \hline
    $\limit_h$   & lower bound on the size of an on-hold structure  & $\Theta(h^{-1})$ & $6h^{-1} + 1$ \\ \hline
    $\Delta_h$   & upper bound on the size of any structure  & $\Theta((h\eps)^{-1})$ & $36(h\eps)^{-1}$ \\ \hline
    \end{tabular}
    \caption{Descriptions of the parameters used in our algorithm.} \label{tab:parameter}
\end{table}

Recall that on \cref*{line:hold} of \algPhase (\cref{alg:phase}), the algorithm stops extending structure whose size reaches $\limit_h$.
However, $\Delta_h$, which denotes the maximum size of a structure in a scale $h$, is not trivially bounded by $O(\limit_h)$.
To see this, notice that although we stop extending a structure $S_\alp$ once its size exceeds $\limit_h$, another structure $S_\beta$ can overtake $S_\alp$.
This can potentially make the size of $S_\beta$ almost $2 \cdot \limit_h$.
Then, another structure $S_\gamma$ can overtake $S_\beta$, making $S_\gamma$
of size almost $3 \cdot \limit_h$, and so on.
In \cref{lem:structure-bound}, we show that this type of growth cannot proceed for ``too long''. 
By leveraging the fact that each free vertex essentially maintains a tree, we obtain an upper bound of $\limit_h \cdot \lmax$ on the structure size.



\subsection{Approximation analysis}
This section proves that \cref{alg:outline} outputs a $(1+\eps)$-approximate maximum matching.
Our analysis starts with the following well-known fact.
 
\begin{definition} [$M$-augmenting $k$-path] For a matching $M$ and an integer $k \geq 1$, define an $M$-augmenting $k$-path as an augmenting path with respect to $M$ consisting of at most $k$ matched edges.
\end{definition}


\begin{restatable}{lemma}{shortpath} \label{lem:short-path} Consider a graph $G$, a matching $M$ of $G$, and a fixed real number $p > 0$. Let $\calP^*$ be the maximum set of disjoint $M$-augmenting $\lmax$-paths in $G$. If $|\calP^*| \leq p|M|$, then $|M|$ is a $(1 + p) \cdot (1 + 1 / \lmax)$-approximate maximum matching.
\end{restatable}
The claims akin to \cref{lem:short-path} are well-established. For the sake of completeness, we provide its proof in \cref{sec:proof-short-path}.

Recall that $\lmax = 3/\eps$.
By letting $p = \frac{1}{\lmax}$ in \cref{lem:short-path}, if at any moment the graph contains less than $p|M|$ disjoint $M$-augmenting $\lmax$-paths, $M$ is smaller than the optimal matching by at most factor of
\[
    \rb{1 + \frac{1}{\lmax}} \cdot \rb{1 + \frac{1}{\lmax}} \leq \rb{1 + \frac{3}{\lmax}} = 1 + \eps.
\]
However, our algorithm does not immediately let $p = 1/\lmax$, but it rather arrives at that fraction gradually by considering different scales. 
The idea behind using varying scales builds on two observations: (1) the larger the value of $p$, the sooner the algorithm stops; (2) the fewer augmenting paths there are, to begin with, the sooner the algorithm stops for a fixed $p$. 
Hence, the algorithm finds a certain fraction of augmentations for a scale $h$, and then a smaller fraction of augmentations remains to be found for $h/2$. We formalize this by the following claim, whose proof is deferred to \cref{sec:proof-active-bound}.
\begin{lemma}[Upper bound on the number of active structures] \label{lem:active-bound} 
Consider a fixed phase for a scale $h$. Let $M$ be the matching at the beginning of the phase. Then, at the end of that phase, there are at most $h|M|$ active structures.
\end{lemma}
We turn \cref{lem:active-bound} into a claim about the matching-size improvement, proved in \cref{sec:proof-path-bound} and formalized as follows:
\begin{lemma} \label{lem:path-bound} Consider a fixed phase in a scale $h$. Let $M$ be the matching at the beginning of the phase. Let $\calP^*$ be the maximum set of disjoint $M$-augmenting $\lmax$-paths in $G$. If $|\calP^*| \geq 4h\lmax |M|$, then at the end of the phase the size of $M$ is increased by a factor of at least $1 + \frac{h\lmax}{\Delta_h}$.
\end{lemma}
It is not hard to build on \cref{lem:path-bound} and obtain a guarantee on the maximum matching approximation after all the phases of a fixed scale.
In particular, for a fixed scale $h$, our algorithm transforms $M$ into a $\rb{1+4h\lmax} \cdot \rb{1+ 1 / \lmax}$-approximate maximum matching, implying that $M$ becomes a $(1 + \eps)$-approximation for scale $h = \Theta\rb{1 / \lmax^2}$.
\begin{lemma}[Matching size at the end of a scale]
\label{lem:scale-bound} At the end of each scale $h$, the matching $M$ is a $\rb{1 + 4h\lmax}\cdot \rb{1 + \frac{1}{\lmax}}$-approximate maximum matching.
\end{lemma}
%
We defer the proof of \cref{lem:scale-bound} to \cref{sec:proof-scale-bound}.

The remaining piece we need for our analysis is an upper bound on the structure size, i.e., $\Delta_h$, which figures in the guarantee of \cref{lem:path-bound}. As discussed in \cref{sec:overview-analysis}, $\Delta_h$ is not trivially upper bounded by $\limit_h$.
Nevertheless, in \cref{sec:proof-structure-bound}, we show that it can only be by a factor of $\lmax$ larger than $\limit_h$.
\begin{lemma}[Upper bound on structure size] \label{lem:structure-bound} 
Consider a fixed phase of a scale $h$. At any moment of the phase, the size of each structure is at most $\limit_h \cdot \lmax$. \end{lemma}

We can now conclude the approximation analysis.
\begin{theorem}
    \cref{alg:outline} outputs a $(1+\eps)$-approximate maximum matching.
\end{theorem}
\begin{proof}
Recall that $\Delta_h = 36(h\eps^{-1})$, which is greater that $\limit_h \cdot \lmax$ (see \cref{tab:parameter}). By \cref{lem:structure-bound}, at any moment of scale $h$, the size of each structure is upper bounded by $\Delta_h$.

By \cref{lem:scale-bound}, at the end of scale $h = \frac{\eps^2}{64} \leq \frac{1}{4\lmax^2}$, $M$ is a $\rb{1 + \frac{1}{\lmax}}\cdot \rb{1 + \frac{1}{\lmax}}$-approximate maximum matching. Since 
$$\rb{1 + \frac{1}{\lmax}}\cdot \rb{1 + \frac{1}{\lmax}} \leq \rb{1 + \frac{3}{\lmax}} = (1 + \eps),$$ 
the output of our algorithm is a $(1 + \eps)$-approximate maximum matching.
\end{proof}

\subsection{Pass and space complexity}
In this section, we prove the following claim.
\begin{lemma}
    \cref{alg:outline} uses $O\rb{1/\eps^6}$ passes and $O\rb{n/\eps^6}$ space.
\end{lemma}

\paragraph{The pass complexity.}
See \cref{tab:parameter} for the parameters used in this algorithm. 
Consider a fixed scale $h$. The number of phases is $\tmax(h) = O(\frac{1}{h \eps})$, each phase executes $\taumax(h) = O(\frac{1}{h \eps})$ $\bundle$s, and each $\bundle$ requires $4$ passes over the stream. Hence, the number of passes in scale $h$ is

\[  
    4 \cdot \tmax(h) \cdot \taumax(h) = 
    4 \cdot \frac{144}{h \eps} \cdot \frac{72}{h \eps} = O\rb{\frac{1}{\eps^2} \cdot \frac{1}{h^2}}.
\]
Since we iterate through $h = \frac{1}{2}, \frac{1}{4}, \frac{1}{8}, \dots, \hmin = \frac{\eps^2}{64}$, the total number of passes is
\[
    O\rb{\frac{1}{\eps^2} \cdot (2^2 + 4^2 + 8^2 + \dots + (64/\eps^2)^2)} = O\rb{\frac{1}{\eps^6}}.
\]
This completes the proof of the pass complexity.

\paragraph{The space complexity.}
Consider a fixed scale $h$.
At any moment in the scale, we need to store the labels of each matched arc and the structure of each free vertex.
Since there are $2|M|$ matched arcs, the labels can be stored using $O(|M|) = O(n)$ space.

Since the structures are vertex-disjoint, the total number of vertices of all structures is bounded by $O(n)$.
Each structure $S_\alp$ may contain up to $O(|V(G_\alp)|^2)$ arcs.
By \cref{lem:structure-bound}, the number of vertices in a structure is upper-bounded by $\Delta_h = O(1 / (h\eps))$.
Hence, the number of arcs in any structure is at most $O(1 / (h\eps)^2)$.
Since there are at most $n$ structures, the arcs of all structures can be stored in $O(n / (h\eps)^2)$ space.

In the execution of a phase, we maintain the set of vertices that are removed so that they can be added back at the end of the phase.
This requires $O(n)$ space.

In $\algCheck$, we compute the set of arcs $A_\alp$ for each free vertex $\alp$.
Recall that $A_\alp$ is the set of all arcs connecting two vertices in $G_\alp$.
Hence, the size of each $A_\alp$ is at most $O(1 / (h\eps^2))$.
As a result, $\algCheck$ can be implemented using $O(n / (h\eps)^2)$ space.
Consequently, a scale $h$ can be implemented using $O(n) + O(n / (h\eps)^2)$ space.
Since the minimum scale is $\hmin = \Theta(\eps^2)$, our algorithm requires $O(n / (\hmin \eps)^2) = O(n / \eps^6)$ space.

\subsection{Proof of \cref{lem:short-path}}
\label{sec:proof-short-path}
For two sets of edges $E_1$ and $E_2$, we use $E_1 \oplus E_2$ to denote the \emph{symmetric difference} of $E_1$ and $E_2$, defined as $(E_1 \cup E_2) - (E_1 \cap E_2)$.
Let $M^*$ be an optimal matching of $G$.
It is known that each component in $M \oplus M^*$ either is an $M$-augmenting path or has the same number of edges from $M$ and $M^*$.
Let $\calP_i$ denote the set of $M$-augmenting $i$-paths in $M \oplus M^*$.

Let $M'$ denote the matching obtained by augmenting $M$ using all paths in $\bigcup_{i \leq \lmax} \calP_i$.
That is, $M' = M \oplus \bigcup_{i \leq \lmax} E(\calP_i)$.
Since $\bigcup_{i \leq \lmax} \calP_i$ is a set of disjoint $M$-augmenting $\lmax$-paths, its size is at most $|\calP^*| \leq p|M|$.
Hence, $|M'| \leq (1 + p)|M|$.
By the definition of $M'$, no component in $M' \oplus M^*$ is an $M'$-augmenting $\lmax$-path.
Hence, if a component in $M' \oplus M^*$ has $i$ edges from $M'$, where $i \geq \lmax + 1$, it has at most
\[
    i+1 = i \cdot \frac{i+1}{i} \leq i \cdot \frac{\lmax+2}{\lmax+1} = i \cdot \rb{1 + \frac{1}{\lmax+1}}
\]
edges from $M^*$. This implies that $|M^*| \leq \rb{1 + \frac{1}{\lmax+1}} \cdot |M'|$.
Hence, we have
\[
    |M^*| \leq \rb{1 + \frac{1}{\lmax+1}} \cdot |M'| \leq \rb{1 + \frac{1}{\lmax+1}} \cdot (1+p) \cdot |M|.
\]
This completes the proof.

\subsection{Upper bound on the number of active structures (Proof of \cref{lem:active-bound})}
\label{sec:proof-active-bound}
Suppose, by contradiction, that there are more than $h|M|$ active structures at the end of a phase.
Then, there are more than $h|M|$ active structures in each $\bundle$ of this phase.
Consider any $\bundle$ $\tau$ in the phase.
At the end of $\tau$, the status of each active structure $S_\alp$ falls into one of the following:
\begin{itemize}
    \item[(1)] $S_\alp$ is on hold,
    \item[(2)] $\alp$ reduces the label of an arc in $\algOvertake$,
    \item[(3)] $\alp$ waits for the next $\bundle$ because part of its structure is overtaken, or
    \item[(4)] $\alp$ contracts a blossom in $S_\alp$.
    \item[(5)] $\alp$ backtracks from the last matched arc in its active path.
\end{itemize}
Each on-hold structure contains at least $\limit_h - 1 = \frac{6}{h}$ matched vertices.
Since the graph contains $2|M|$ matched vertices in total, there are at most $2|M|/(\limit_h - 1) = h|M|/3$ on-hold structures.
Therefore, at least $2h|M|/3$ free vertices satisfy one of (2)-(5).

For each $\alp$ satisfying (2) or (4), a label of a matched arc is reduced by at least one.
For each $\alp$ satisfying (5), a free vertex backtracks from a matched arc whose label was reduced previously.
Each free vertex in (2) causes at most one free vertex to satisfy (3).
Therefore, in $\tau$, there are at least $h|M|/3$ free vertices satisfying (2), (4), or (5).

In a phase, the label of each matched arc is reduced at most $\lmax + 1$ times.
Since there are $2|M|$ matched arcs, the total number of label changes in the phase is at most $2 (\lmax + 1) |M| \leq 4 \lmax |M|$, and the total number of backtracking operations is thus at most $4 \lmax |M|$.
This shows that the total number of $\bundle$s is less than

\[ \frac{8 \lmax |M|}{h|M| / 3} = \frac{24 \lmax}{h}, \]
which contradicts the fact that the algorithm executes $\taumax(h) = \frac{24 \lmax}{h}$ $\bundle$s in this phase. Hence, the lemma holds.

\subsection{Proof of \cref{lem:path-bound}}
\label{sec:proof-path-bound}
Consider an augmenting path $P$ in $\calP^*$. By \cref{lem:active}, at the end of the phase, either some vertices in $P$ are removed, or $P$ contains a critical arc or a critical vertex.
By \cref{lem:active-bound}, at the end of the phase, there are at most $h|M|$ active structure.
Clearly, each active structure $S_\alp$ contains exactly one critical vertex, that is, the vertex $\alp$.
Therefore, each active structure contains one critical vertex and, at most, $\lmax$ critical arcs.
That is, at most $h|M| \cdot (\lmax + 1)$ paths in $\calP^*$ contain a critical arc or a critical vertex; recall that all the paths in $\calP^*$ are vertex-disjoint.

We proceed to bound the number of paths in $\calP^*$ containing a removed arc or vertex.
Let $\calP$ be the set of disjoint $M$-augmenting paths found by $\algPhase$ in this phase.
When an augmenting path between two free vertices $\alp$ and $\beta$ is found, our algorithm also removes $S_\alp$ and $S_\beta$.
Each removed structure contains at most $\Delta_h$ vertices.
Hence, at most $|\calP| \cdot 2\Delta_h$ paths in $\calP^*$ contain a removed matched arc or removed free vertex.
Consequently, we obtain $|\calP^*| \leq h|M|(\lmax + 1) + |\calP| \cdot 2\Delta_h$.
Assume that $|\calP^*|$ contains more than $4h|M|\lmax$ paths, then we have
\[
    |\calP|
    \geq \frac{|\calP^*| - h|M|(\lmax + 1)}{2\Delta_h}
    \geq \frac{2h\lmax}{2\Delta_h}|M|
    = \frac{h\lmax}{\Delta_h}|M|.
\]
Our algorithm augments $M$ by using the augmenting paths in $\calP$ at the end of the phase.
Hence, the size of $M$ is increased by a factor of $(1 + \frac{h\lmax}{\Delta_h})$ at the end of this phase.
This completes the proof of the claim.

\subsection{Matching size at the end of a scale (Proof of \cref{lem:scale-bound})}
\label{sec:proof-scale-bound}
We prove this lemma by induction.

\paragraph{Base case:} 
At the beginning of our algorithm, $M$ is initialized as a 2-approximate maximum matching. The size of $M$ never decreases throughout the execution. In the first scale, $h = \frac{1}{2}$. 
Since $\rb{1 + 4h\lmax} \cdot \rb{1 + \frac{1}{\lmax}} \geq 2$, $M$ is indeed a $\rb{1 + 4h\lmax} \cdot \rb{1 + \frac{1}{\lmax}}$-approximate maximum matching at the end of $h$.
    
    
\paragraph{Inductive case:}    
    Let $h < \frac{1}{2}$ be a scale.
    Assume that the lemma holds for the scale $2h$.
    Then, $M$ is a $(1 + 8h\lmax) \cdot \rb{1 + \frac{1}{\lmax}}$-approximate maximum matching at the beginning of scale $h$ by the inductive hypothesis.
    
    Suppose, toward a contradiction, that $M$ is not a $\rb{1 + 4h\lmax} \cdot \rb{1 + \frac{1}{\lmax}}$-approximate maximum matching at the end of scale $h$.
    Since the size of $M$ never decreases, $M$ is not a $\rb{1 + 4h\lmax} \cdot \rb{1 + \frac{1}{\lmax}}$-approximate maximum matching at the beginning of each phase in $h$.
    
    By setting $p = 4h\lmax$ in \cref{lem:short-path}, at the beginning of each phase in $h$, $G$ contains at least $4h\lmax$ $M$-augmenting $\lmax$-paths.
    By \cref{lem:path-bound}, after each phase, $|M|$ is increased by a factor of $\frac{h\lmax}{\Delta_h}$.
    Let $T = \frac{144}{h \eps}$ be the number of phases in the scale.
    Since $T > 4\Delta_h$, $|M|$ is increased by a factor of
    \[
        \rb{1 + \frac{h\lmax}{\Delta_h}}^T
        \geq 1 + T \frac{h\lmax}{\Delta_h}
        > 1 + 4h\lmax,
    \]
    where the first inequality comes from the fact that $(1 + x)^k \geq 1 + kx$ for all positive integer $k$ and positive real number $x$.\footnote{This inequality can be proven by an induction on $k$. In the inductive argument, we use the fact that $(1 + (k-1)x)(1+x) \geq 1 + kx$.}
    
    Let $M^*$ be the maximum matching.
    Recall that $M$ is a $(1 + 8h\lmax)\cdot \rb{1 + \frac{1}{\lmax}}$-approximate maximum matching at the beginning of $h$.
    At the end of $h$, $|M^*| / |M|$ is at most
    \[
        (1 + 8h\lmax)\cdot \rb{1 + \frac{1}{\lmax}} / \rb{1 + 4h\lmax} \leq \rb{1 + 4h\lmax}\cdot \rb{1 + \frac{1}{\lmax}},
    \]
    which contradicts the assumption that $M$ is not a $\rb{1 + 4h\lmax}\cdot \rb{1 + \frac{1}{\lmax}}$-approximate maximum matching at the end of $h$.
    This completes the proof.

\subsection{Upper bound on structure size (Proof of \cref{lem:structure-bound})}
\label{sec:proof-structure-bound}
Consider a matched arc $a' = (u', v')$ of $G'$. We define $T'(a')$ as follows.
If $a'$ is in some structure $S_\alp$, $T'(a')$ is the subtree of $T'_\alp$ rooted at $u'$;
otherwise, $T'(a')$ is defined as the tree consisting of a single arc $a'$.
Define $size(a')$ as the number of vertices $v \in G_\alp$ such that $\Omg(v) \in T'(a')$. That is, $size(a')$ is the number of vertices of $G$ in $T'(a')$, including those vertices of $G$ representing contracted blossoms of $T'(a')$.

To prove this lemma, we first show the following invariant:
\begin{invariant}
    At any moment of the phase, $size(a') \leq (\lmax - \ell(a'))\limit_h + 1$ holds for each matched arc $a'$ in $G'$.    
\end{invariant}

\paragraph{Proving that the invariant holds.}
For each unvisited matched arc $a'$, it holds that $size(a') = 1$.
Hence, the invariant holds at the beginning of the first $\bundle$.
Suppose, by contradiction, that the invariant is violated at some moment in the phase.
For a matched arc $a' \in E(G')$ that is contained in a structure $S_\alp$, there are only four types of events that may change $size(a')$:
\begin{enumerate}
    \item $S_\alp$ overtakes a matched arc and potentially a part of another structure,
    \item part of $S_\alp$ is overtaken,
    \item $\algAugment$ is invoked to remove $S_\alp$,
    \item $\algContract$ is invoked to contract a blossom in $S_\alp$.
\end{enumerate}
Consider the first event after which a matched arc $a' \in E(G')$ has $size(a') > (\lmax - \ell(a'))\limit_h + 1$.
Right before the event, $a'$ must be in a structure $S_\alp$, for otherwise $size(a') = 1$ after the event.
We claim that the increase of $size(a')$ must be because $S_\alp$ overtakes a part of another structure.
Suppose this is not true.
If part of $S_\alp$ is overtaken, $size(a')$ cannot increase.
If $\algAugment$ removes $S_\alp$, $a'$ is also removed and thus no longer a matched arc after the event, which contradicts the definition of $a'$.
If $\algContract$ contracts a blossom in $S_\alp$, then either $a'$ is contracted (which is a contradiction) or $size(a')$ is unchanged.
Hence, the increase of $size(a')$ is due to an overtaking operation performed by $S_\alp$.

Let $a^* \in E(G')$ be the matched arc overtaken by $S_\alp$ in the event.
Since $size(a')$ increases in the event, the overtaking operation must move the subtree $T'(a^*)$ under the subtree of $T'(a')$.
This implies that $\ell(a^*) > \ell(a')$ before the event.
Consider the moment right before the event.
Since the invariant holds, we have $size(a^*) \leq (\lmax - \ell(a^*))\limit_h + 1$.
Since $S_\alp$ is not on hold and not modified, $size(a') \leq \limit_h - 1$.
Hence, after the overtaking operation, $size(a')$ becomes at most
\[
(\lmax - \ell(a^*))\limit_h + 1 + (\limit_h - 1) < (\lmax - \ell(a') - 1)\limit_h + \limit_h \leq (\lmax - \ell(a'))\limit_h,
\]
which contradicts the definition of $a'$.
Therefore, the invariant holds.

\paragraph{From the invariant to the structure size.}
We proceed to show that the size of each structure is at most $\lmax \cdot \limit_h$.
Towards a contradiction, assume that $S_\alp$ has a structure size greater than $\lmax \cdot \limit_h$, and consider the first moment when it happens.

Using a similar argument to the one for showing the invariant above, it follows that $S_\alp$ overtakes a matched arc $a^*$.
Right before the overtaking, the size of $S_\alp$ is at most $\limit_h - 1$ because it is not on hold and not modified;
furthermore, by the invariant,
\[
    size(a^*) \leq (\lmax - \ell(a^*))\limit_h + 1 \leq (\lmax - 1)\limit_h + 1.
\]
Thus, after the overtaking operation, the size of $S_\alp$ is at most 
\[
(\lmax - 1)\limit_h + 1 + (\limit_h - 1) \leq \lmax \cdot \limit_h,
\]
which contradicts the definition of $\alp$.
Consequently, the lemma holds.

\bibliographystyle{alpha}
\bibliography{reference}


\appendix

\section{Proof of \cref{lem:inv-preservation}} \label{sec:appendix-inv}

In $\algPhase$, there are four procedures that can modify the structures:
$\algAugment$, $\algContract$, $\algOvertake$, and $\algBacktrack$.
We prove that these procedures preserve all properties of a structure, which includes the disjointness, tree representation, and unique arc properties in \cref{def:structure}.
We also show that \cref{inv:increasing-labeling} is preserved by the procedures.

$\algAugment$ removes both structures that it modifies;
$\algBacktrack$ only changes the working vertex of each structure;
Therefore, they clearly preserve the invariant and all properties of a structure.
In the following, we argue that $\algContract$ and $\algOvertake$ also preserve these properties.

\begin{lemma} \label{lem:appendix1} Let $(g, a, k)$ be the input to an invocation of $\algOvertake$.
After the invocation, the updated $S_\alp$ is a structure per \cref{def:structure}.
If $a \in S_\beta$ for $\alp \neq \beta$, after the same invocation, $S_\beta$ is
also a structure per \cref{def:structure}.
Moreover, if \cref{inv:increasing-labeling} holds before this invocation, it holds after the invocation as well.
\end{lemma}
\begin{proof}
Let $g = (u, v)$ and $a = (v, t)$.
Without loss of generality, we assume that $a$ is in a structure $S_\beta \neq S_\alp$.
The overtaking operation moves the subtree of $\Omg(v)$ from $T'_\beta$ to $T'_\alp$.
We denote by $T''$ the subtree moved in the operation.
Define $p$ as in \cref{sec:overtake};
that is, $p \in G$ is the vertex such that $(\Omg(p), \Omg(v))$ is the unique unmatched arc connecting $\Omg(v)$ with its parent before the invocation of $\algOvertake$.

\paragraph{$G_\alp$ is a subgraph, $\Omg_\alp$ is a regular set of blossoms, and $w'_\alp$ is a vertex of $T'_\alp$:}
Clearly, $G_\alp$ is a subgraph after the invocation.
Before the overtaking operation, the root of $T''$ was the inner vertex $\Omg(v)$.
After the operation, it becomes a child of the outer vertex $\Omg(u)$.
Hence, each inner vertex (resp. outer vertex) of $T''$ remains to be an inner vertex (resp. outer vertex) after the operation.
By the observation above, each inner vertex of $T'_\alp$ is a trivial blossom after the operation.
Consequently, $\Omg_\alp$ is a regular blossom after the invocation.
After the invocation of $\algOvertake$, $w'_\alp$ is either set to be $\Omg(t)$ or $w'_\beta$
(that is, the working vertex of $S_\beta$ before the invocation).
In either case, $w'_\alp$ is an outer vertex of $T'_\alp$.

\paragraph{$G_\beta$ is a subgraph, $\Omg_\beta$ is a regular set of blossoms, and $w'_\beta$ is a vertex of $T'_\beta$:}
Clearly, $G_\beta$ is a subgraph after the invocation.
Since $\algOvertake$ only removes the rooted subtree $T''$ from $T'_\beta$, $\Omg_\beta$ is a regular blossom after the invocation.
In the overtaking operation, $w'_\beta$ is either unchanged or updated as $\Omg(p)$.
In either case, it is an outer vertex of $T'_\beta$ after the invocation.

\paragraph{The disjointness property holds:}
Since $\algOvertake$ only moves a set of vertices and arcs from $S_\alp$ to $S_\beta$, the disjointness property holds after the invocation.

\paragraph{The tree representation and unique arc properties hold :}
The first step of the overtaking operation is to remove $(p, v)$ from $S_\beta$ and add $(u, v)$ to $S_\alp$.
By the unique arc property, after $(p, v)$ is removed, the subtree $T''$ is disconnected from $T'_\beta$.
By adding $(u, v)$ to $S_\alp$, we ensures that $T'_\alp$ contains an arc from the new parent of $\Omg(v)$ (that is, $\Omg(u)$) to $\Omg(v)$.
In addition, since $(u, v)$ is the only arc in $G_\alp$ that is adjacent to $T''$, we know that no other arc $(x, y) \in G_\alp$ corresponds to an arc between $\Omg(u)$ and $\Omg(v)$.
Hence, both the tree representation and unique arc properties hold for $T'_\alp$ after the invocation.

Since $\algOvertake$ only moves a subtree from $S_\beta$ to $S_\alp$, the two properties also hold for $S_\beta$ after the invocation.

\paragraph{\cref{inv:increasing-labeling}:}
For ease of presentation, we say an alternating tree $T^*$ of $G'$ is \emph{properly labeled} if the following holds:
For each root-to-leaf path $(a'_1, a'_2, \dots, a'_k)$ on $T^*$, it holds that $\ell(a'_1) < \ell(a'_2) < \dots < \ell(a'_k)$.
To prove \cref{inv:increasing-labeling}, it suffices to show that $T'_\alp$ and $T'_\beta$ are properly labeled after the invocation.

We first argue that $T'_\alp$ is properly labeled.
Since \cref{inv:increasing-labeling} holds before the invocation, $T''$ is properly labeled.
Suppose that $\Omg(u)$ is the root of $T'_\alp$.
Then, $\algOvertake$ only attaches $T''$ as a child of the root.
Since $T''$ is properly labeled, it is easy to see that $T'_\alp$ is also properly labeled after the addition.

Suppose that $\Omg(u)$ is not the root of $T'_\alp$.
In this case, $\algOvertake$ add attaches $T''$ as a child of the non-root vertex $\Omg(u)$.
Let $a'$ and $a''$ denote, respectively, the matched arc adjacent to $\Omg(u)$ and $\Omg(v)$.
Since $T''$ is properly labeled and its root $\Omg(v)$ only has one child (because $\Omg(v)$ is an inner vertex), it suffices to show that $\ell(a') < \ell(a'')$ after the invocation.
Recall that $\algOvertake$ is only invoked in the execution of $\algExtend$.
As defined in \cref{sec:extend}, The input $k$ to $\algOvertake$ is computed as $\ell(a') + 1$.
Hence, it holds that $k > \ell(a')$.
Since $\ell(a'')$ is updated as $k$, we know that $\ell(a') < \ell(a'')$ after the invocation.
Consequently, $T'_\alp$ is properly labeled after the invocation.

We proceed to argue that $T'_\beta$ is properly labeled after the invocation.
By \cref{inv:increasing-labeling}, every subtree of $T'_\beta$ is properly labeled before the invocation.
In the overtaking operation, only the rooted subtree $T''$ is removed from $S_\beta$.
Therefore, it is easy to see that $T'_\beta$ remains properly labeled after the removal.
Since both $T'_\alp$ and $T'_\beta$ are properly labeled after the invocation, \cref{inv:increasing-labeling} is preserved.

In the argument above, it has been assumed that $a$ is in a structure $S_\beta$.
We remark that if $a$ is not in any structure, we may as well think of the operation as overtaking a subtree consisting of a single matched arc from a structure $S_\beta$.
Hence, a similar argument can be applied for the case.
\end{proof}

\begin{lemma} \label{lem:appendix2} Let $g = (u, v)$ be the input to an invocation of $\algContract$.
After the invocation, the updated $S_\alp$ is a structure per \cref{def:structure}.
Moreover, if \cref{inv:increasing-labeling} holds before this invocation, it holds after the invocation as well.
\end{lemma}
\begin{proof}
Let $B$ denote the blossom contracted in the invocation.
Let $T'_1$ and $T'_2$ denote, respectively, the tree $T'_\alp$ before and after the invocation.
Let $\Omg_1$ and $\Omg_2$ denote, respectively, the set $\Omg_\alp$ before and after the invocation.
Let $u'$ and $v'$ denote, respectively $\Omg_1(u)$ and $\Omg_1(v)$.

\paragraph{$G_\alp$ is a subgraph, $\Omg_\alp$ is a regular set of blossoms, and $w'_\alp$ is a vertex of $T'_\alp$:}
Since $\algContract$ does not change the vertex set of $G_\alp$, $G_\alp$ is a subgraph after the invocation.
By \cref{lem:contraction}, $T'_2$ is an alternating tree containing $B$ as an outer vertex;
the other inner and outer vertices of $T'_2$ are those of $T'_1$ which are not in $B$.
Therefore, $\Omg_\alp$ is a regular set of blossoms after the invocation.
Since $B$ is set to be the working vertex, $w'_\alp$ is an outer vertex of $T'_\alp$ after the invocation.

\paragraph{The disjointness property holds:} Since $\algContract$ does not change the vertex set of $G_\alp$, the disjointness property holds after the invocation.

\paragraph{The tree representation and unique arc properties hold:}
To prove the property, we first show that each arc $(x', y')$ of $T'_2$ can be mapped to a non-blossom arc $(x, y)$ in $G_\alp$.

Since the arcs of $T'_2$ correspond to the arcs of $T_1$ that are not contracted, there is an injection $f_1$ from $E(T'_2)$ to $E(T'_1)$.
The injection is formally defined as follows.
Consider an arc $(x', y')$ in $T'_2$.
If $x', y' \neq B$, then $T'_1$ also contains the arc $(x', y')$, and we map $(x', y')$ to itself.
If $x' = B$, then the parent $p'$ of $y'$ on $T'_1$ was contracted in the invocation; in this case, we map $(x', y')$ to $(p', y')$.
Suppose that $y' = B$.
Since $B$ is an outer vertex of $T'_2$, $x'$ is an inner vertex in $T'_2$.
By \cref{lem:contraction}, $x'$ is also an inner vertex in $T'_1$.
Let $c'$ be the only child of $x'$ in $T'_1$.
In this case, we map $(x', y')$ to $(x', c')$.

Since the two properties hold before the invocation, there is a bijection $f_2$ from the arcs in $T'_1$ to the non-blossom arcs (with respect to $\Omg_1$) in $G_\alp$.
The composition, $f_2 \circ f_1$, of the two mappings is an injection from the arcs of $T'_2$ to the non-blossom arcs (with respect to $\Omg_1$) in $G_\alp$.
It is not hard to verify that $f_2 \circ f_1$ maps an arc $(x', y') \in T'_2$ to the unique arc $(x, y) \in G_\alp$ such that $\Omg_2(x) = x'$ and $\Omg_2(y) = y'$.
Hence, the two properties hold after the invocation.





\paragraph{\cref{inv:increasing-labeling} holds:}
Consider an alternating tree.
We say a matched arc $(u', v')$ in the tree is \emph{above} another matched arc $(x', y')$ if $v'$ is an ancestor of $x'$.

Observe that the contraction of $B$ preserves the ancestor-descendant relationships;
more precisely, if a matched arc $a'_1$ is above a matched arc $a'_2$ in $T'_2$, then $f_1(a'_1)$ is also above $f_1(a'_2)$ in $T'_1$.
Consider a root-to-leaf path $(a'_1, a'_2, \dots, a'_k)$ on $T'_2$.
By the observation above, for each $i \leq k$, $f_1(a'_i)$ is above $f_1(a'_{i+1})$ in $T'_1$.
Since $\algContract$ does not change the label of each arc that is not contracted, we have $\ell(a'_i) = \ell(f_1(a'_i))$ for all $1 \leq i \leq k$.
Since \cref{inv:increasing-labeling} holds before the invocation, we know that $\ell(a'_1) < \ell(a'_2) < \dots < \ell(a'_k)$.
Since the argument above holds for all root-to-leaf paths on $T'_2$, \cref{inv:increasing-labeling} is preserved.
\end{proof}

\noindent \cref{lem:appendix1} and \cref{lem:appendix2} complete the proof of \cref{lem:inv-preservation}.


\end{document}